\documentclass[10pt,journal,letterpaper]{IEEEtran}
\IEEEoverridecommandlockouts
\usepackage{amsmath,amsfonts}
\usepackage{array}
\usepackage{textcomp}
\usepackage{stfloats}
\usepackage{url}
\usepackage{verbatim}
\usepackage{graphicx}
\usepackage{cite}
\usepackage[noend]{algpseudocode}
\usepackage[linesnumbered,ruled]{algorithm2e}
\usepackage{upgreek}
\usepackage{subfigure}
\usepackage{color}
\def\BibTeX{{\rm B\kern-.05em{\sc i\kern-.025em b}\kern-.08em
    T\kern-.1667em\lower.7ex\hbox{E}\kern-.125emX}}
    
\usepackage{amsthm}
\usepackage{amssymb}
\usepackage{makecell}
\usepackage{cleveref}
\usepackage{enumerate}
\newtheorem{theorem}{Theorem}
\newtheorem{definition}{Definition}
\theoremstyle{definition}
\newtheorem{remark}{Remark}
\theoremstyle{proposition}
\newtheorem{proposition}{Proposition}
\theoremstyle{Assumption}

\theoremstyle{Lemma}

\newtheorem{corollary}{\textbf{Corollary}}
\graphicspath{{figures/}}
\usepackage{enumitem}

\newcommand{\myref}[1]{%
    \ifthenelse{\equal{#1}{proof_proposition_1}}{A}{%
    \ifthenelse{\equal{#1}{proof_lemma_0}}{B}{%
    \ifthenelse{\equal{#1}{proof_lemma_2}}{C}{%
    \ifthenelse{\equal{#1}{proof_remark_2}}{D}{%
    \ifthenelse{\equal{#1}{proof_theorem_3}}{E}{%
    \ifthenelse{\equal{#1}{proof_theorem_4}}{F}{%
    {\textbf{??}}}}}}}}%
}

\ifodd 0
\definecolor{darkgreen}{RGB}{0,200,0}
\else

\fi
\begin{document}
\title{
SLIDE: Simultaneous Model Downloading and Inference at the Wireless Network Edge   \\
\thanks{Guanqiao Qu, Tao Li, Qian Chen, and Xianhao Chen are with the Department of Electrical and Computer Engineering, The University of Hong Kong, Hong Kong SAR, China. (e-mail: gqqu@eee.hku.hk; litao@eee.hku.hk; qchen@eee.hku.hk; xchen@eee.hku.hk. Sheng Zhou is with the Department of Electronic Engineering, Tsinghua University, Beijing, China. (e-mail: sheng.zhou@tsinghua.edu.cn). The work was supported in part by the Research Grants Council of
Hong Kong under Grant 27213824 and CRS HKU702/2. \textit{(Corresponding author: Xianhao Chen.)}

}
}
\author{Guanqiao Qu,~\IEEEmembership{Graduate Student Member,~IEEE}, Tao Li, Qian Chen,~\IEEEmembership{Member,~IEEE}, \\Xianhao Chen,~\IEEEmembership{Member,~IEEE}, Sheng Zhou,~\IEEEmembership{Senior Member,~IEEE}}

\maketitle

\begin{abstract}
To support on-device inference, the next-generation mobile networks are expected to support real-time model downloading services to mobile users. However, powerful AI models typically have large model sizes, resulting in excessive end-to-end (E2E) downloading-and-inference (DAI) latency. To address this issue, we propose a \underline{s}imultaneous mode\underline{l} download\underline{i}ng an\underline{d} inf\underline{e}rence (SLIDE) framework, which allows users to perform inference with downloaded layers while simultaneously receiving the remaining layers of the model. To this end, we formulate a task throughput maximization problem by jointly optimizing model provisioning, spectrum bandwidth allocation, and computing resource allocation for multi-user downlink systems. Unlike traditional DAI frameworks, SLIDE introduces recursive dependencies across layers, where inference latency depends recursively on the downloading bandwidth and computing resource allocation for each of the preceding layers. To solve this challenging problem, we design an efficient algorithm that acquires the optimal solution with polynomial-time complexity. Simulation results demonstrate that the proposed SLIDE framework significantly improves task throughput under latency and communication resource constraints compared with the conventional model downloading schemes.
\end{abstract}

\section{Introduction}

With mounting concerns about privacy and the increasing computing power of mobile devices, on-device inference has become a prevalent paradigm, enabling users to run AI models directly on their end devices~\cite{10.1145/3450268.3453524,10631278}. In this respect, \textit{AI model downloading} has emerged as an essential service in next-generation mobile networks to support ubiquitous on-device inference~\cite{3gpp.22.874,mao2023green}. 
Unlike the traditional model (or mobile app) downloading, which is often time-consuming, the future model downloading is expected to be a real-time service~\cite{10183793}. For instance, the 5G standard specifies model downloading latency requirements as low as 1 second for time-sensitive applications such as autonomous driving and robotics~\cite{3gpp.22.874}. \textit{Real-time} model downloading benefits on-device inference in the following aspects. (i) Storage efficiency: Due to the diverse range of AI applications and the large size of models, users prefer not to, or simply cannot, pre-store all the AI models that they might need in advance. For example, Google's on-device language model Gemini Nano-2 remains 1.51 GB even after 4-bit quantization~\cite{team2023gemini}. Real-time model downloading enables users to fetch models \textit{on demand} from off-site locations, thereby saving on-device storage. (ii) Just-in-time intelligence: Real-time downloading is crucial for time-sensitive applications when \textit{pre-downloading is infeasible}, particularly when user demands are unpredictable or specific model versions are not available in advance. For example, mobile networks are expected to deliver region-specific models to users, such as localized models for augmented reality~\cite{9162277} and autonomous driving~\cite{ma2025sense4fl}. In such cases, since the downloading of the context-aware model can occur only when users enter the region, real-time model downloading is essential to ensure low service latency for users. Moreover, while recent advances in large language models (LLMs) have promoted the deployment of on-device foundation models that can be reused across multiple applications, such a paradigm does not eliminate the need for real-time model parameter downloading. In practice, since foundation models face inherent tradeoffs between model size and generality, they cannot be optimally specialized for all tasks under limited resource budgets. Therefore, users often request task-specific parameters (such as LoRA matrices~\cite{hu2021lora}) for their requested tasks. As a result, dynamic and real-time model downloading is essential for providing flexible, timely, and up-to-date on-device intelligence to end users.

\textbf{Challenges}: However, while a model is nothing but just one type of digital content, existing content downloading schemes can hardly fulfill the real-time requirements of model downloading. Specifically, all prior content transmission schemes, including model downloading schemes~\cite{wu2023efficient}, focus solely on data transmission without considering \textit{downloading-inference co-design}. Nonetheless, the end-to-end (E2E) latency is what really matters. Imagining that a vehicle enters an area requesting a region-specific navigation model from a base station~\cite{hu2025agentscomerge}, the service cannot be fulfilled until an entire \textit{download-and-inference} (DAI) process has been accomplished. In practical wireless networks, DAI entails significant model downloading latency before inference, leading to excessive E2E service latency, as echoed by 3GPP findings~\cite{3gpp.22.874}. Considering the typical 5G download speed of 170.1 Mbps~\cite{cmhk5g}, downloading a ResNet-18 model takes about 2.1 seconds, and inference using such a ResNet-18 on a Raspberry Pi 4 with two images takes about 2 seconds~\cite{10.1145/3452411.3464446}. As a result, the separated design of downloading and inference in DAI may render E2E latency intolerant to time-sensitive applications.

\textbf{Our solution}: To reduce E2E service latency, our key idea lies in overlapping communication and computing in DAI, called \underline{s}imultaneous mode\underline{l} download\underline{i}ng an\underline{d} inf\underline{e}rence (SLIDE). Unlike the conventional DAI paradigm that initiates inference only after the entire model is downloaded, our framework enables users to start inference with the downloaded layers while simultaneously receiving the remaining layers. The advantages of this framework are evident. In Fig. \ref{fig:intro_comp}, we demonstrate the performance gain of SLIDE by comparing the E2E latency for different approaches, including: (i) \textit{disk-loading inference}, where the model is stored on the disk of the user device and is loaded from the disk into system memory and then into GPU memory for inference~\cite{zhu2023fastdimenet++}; 
(ii) \textit{memory-preloading inference}, where the model is already in system memory due to page cache or prior requests and can be directly loaded into GPU memory for inference~\cite{10.1145/3698038.3698509}; 
and (iii) \textit{conventional DAI}, where inference begins after the model has been fully downloaded from a base station (BS)~\cite{wu2023efficient}. 
As shown in Fig. \ref{fig:intro_comp}, the proposed SLIDE significantly reduces the E2E latency by 32.5\% on average compared with the conventional DAI, while being only 0.2 times slower than the disk-loading inference approach!


\begin{figure}[t]
\centering
\includegraphics[width=0.3\textwidth]{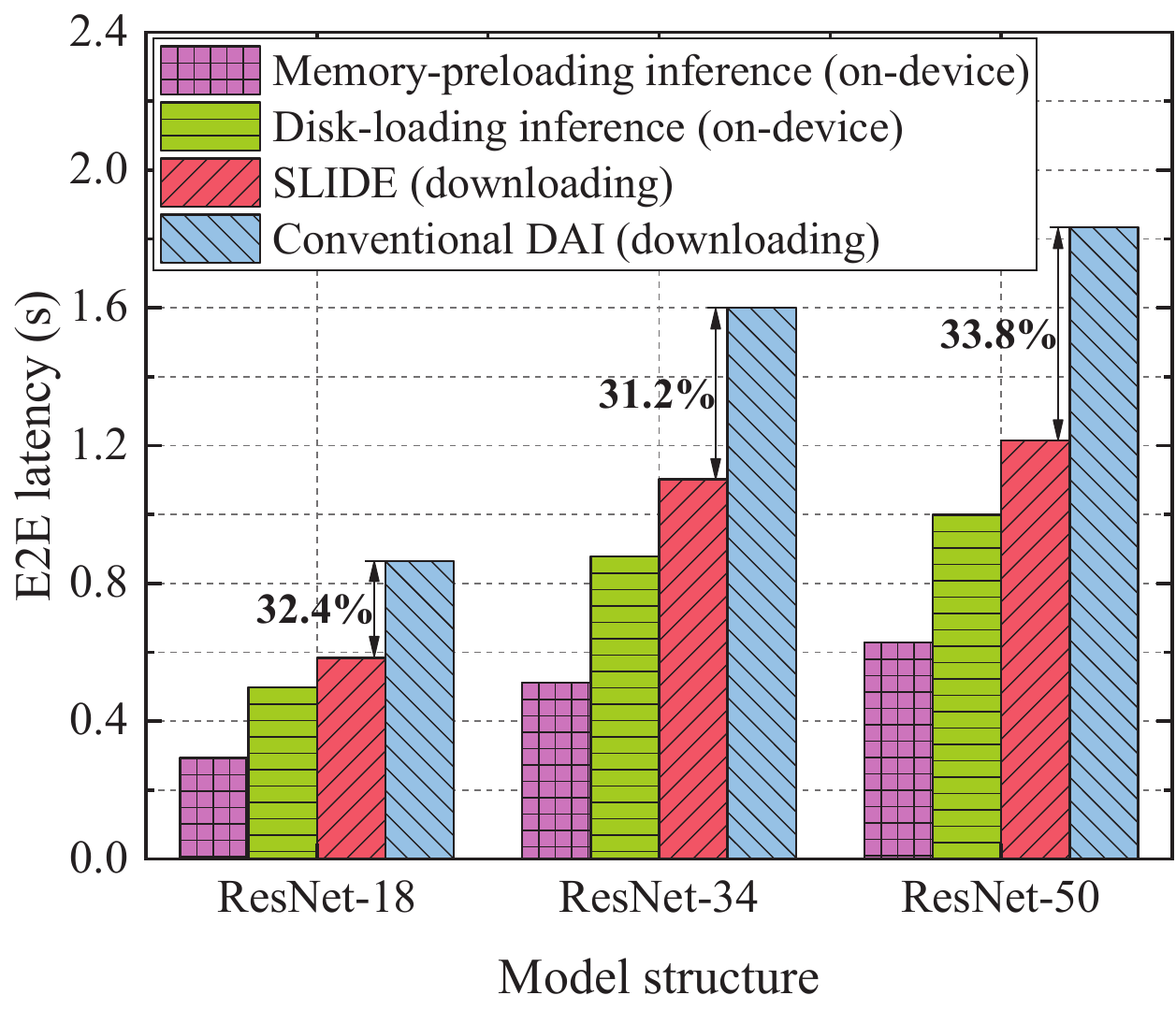}
\vspace{-0.25cm}
\caption{The E2E latency in wireless networks, where an end user downloads an AI model from a BS to perform local inference. The E2E latency comprises downloading and inference latency. 
We assume the BS downloads models with a transmit power spectral density of -29 dBm/Hz~\cite{3gpp.38.104} over a 30-MHz channel, and the distance is set to 100 m. The inference is executed on a Jetson Orin Nano with a GPU frequency of 624.75 MHz, using the CIFAR-10 dataset~\cite{krizhevsky2009learning}, a batch size of 1, and models from the ResNet family~\cite{he2016deep}. 
}
\vspace{-10pt}
\label{fig:intro_comp}
\end{figure}
With the innovative SLIDE paradigm, in this paper, we consider a multi-user downlink system where a BS delivers models to users for on-device inference with limited spectrum bandwidth. We formulate a joint model provisioning, spectrum bandwidth allocation, and computing resource allocation problem to maximize the number of completed AI tasks (task throughput) under E2E deadline constraints. The problem turns out to be a mixed-integer nonlinear programming problem. Note that the optimization problem is highly challenging, as the E2E latency in SLIDE follows a recursive expression across layers: Inference on a layer starts only after downloading the current layer and completing inference of the previous layer. Despite the challenges, we develop an efficient algorithm to find the optimal solution to SLIDE with polynomial time complexity. Concretely, we first decouple the computing resource allocation from the original joint problem. For any fixed bandwidth allocation, we then determine the corresponding model provisioning and computing resource allocation. Subsequently, we obtain the minimum feasible bandwidth allocation for each user under latency constraints, based on which we determine the optimal solution to maximize the number of served users. The contributions of this paper are summarized as follows.

\begin{enumerate}
    \item We propose the SLIDE framework, which enables users to start inference with the downloaded AI model layers while simultaneously receiving the remaining layers. Moreover, we formulate an optimization problem of model provisioning, spectrum bandwidth allocation, and computing resource allocation in a multi-user downlink wireless system to maximize the number of completed tasks under the latency requirements of users. The problem is a mixed-integer nonlinear programming problem.
    \item To facilitate efficient solution finding, we first derive the model provisioning and computing resource allocation for each user under a given spectrum bandwidth allocation. Then, we show that solving the original problem is equivalent to solving the spectrum bandwidth allocation problem based on the model provisioning and computing resource allocation decisions. Finally, we propose an iterative algorithm that leverages the problem properties to derive the minimum feasible bandwidth allocation for each user, followed by an efficient algorithm that incrementally constructs the optimal solution to the original problem. 
    \item Our experiments have demonstrated that the proposed SLIDE significantly boosts the system performance compared with conventional AI model downloading schemes.
\end{enumerate}

The rest of this paper is organized as follows. Section II reviews the related work. Section III introduces the proposed SLIDE framework and formulates the task throughput maximization problem. In Section IV, we present an algorithm to obtain the optimal solution. Section V provides the numerical results, and Section VI concludes this paper. 


\section{Related Work}

AI model downloading is an emerging field, attracting growing attention. We review the relevant resource allocation methods below.

\textbf{Model downloading}: In existing DAI frameworks, users download the required AI models from edge servers (e.g., BSs) for local inference~\cite{qu2024trimcachingICDCS}. To reduce model downloading latency, Wu et al. leveraged parameter sharing across models to enable the BS to broadcast shared model blocks to users~\cite{wu2023efficient}. Moreover, to better utilize computing resources of both user devices and edge servers, cooperative edge inference splits AI models into user-side and server-side sub-models~\cite{9843917}, 
allowing users to download and run the user-side sub-model and upload intermediate features to the server for further inference~\cite{9758628,9093772}. In these studies, on-device inference begins only after model downloading completes, which can hence be reduced to conventional resource allocation problems.
However, these approaches fall short for SLIDE since they do not consider the overlapping of model downloading and local inference. This overlapping introduces recursive dependencies across layers in the E2E latency of SLIDE, making the latency jointly influenced by bandwidth allocation, layer sizes, and computing workload. This indicates that SLIDE requires a tailored resource allocation strategy.

The idea of ``\textit{simultaneous delivery and processing}", similar to that of SLIDE, has been explored in prior work, yet in \textit{different fields}. We elaborate on why they fall short of supporting SLIDE at the wireless edge as follows.


\textbf{Video streaming}: Video streaming techniques divide videos into multiple segments, allowing users to download the video segments and consume the video simultaneously~\cite{seufert2014survey,5990550}.  
This progressive downloading enables users to start playback before downloading the full video. However, this approach is inapplicable to downloading and inference. Although AI models can similarly be divided into layers, users can only benefit from the models upon inference completion for the entire model. Moreover, unlike video segments, which are uniformly timed~\cite{10.1145/2557642.2557658,yin2015control}, computing workloads and data sizes vary across layers. These factors prevent the application of video streaming schemes to SLIDE.

\textbf{Layer-wise inference}: In computer architecture, layer-wise inference sequentially transfers model layers from disk/system memory to GPU memory, allowing the GPU to perform inference with each layer once being loaded~\cite{peng2024harnessing,bhattacharya2016sparsification}. Similar to SLIDE, loading of the subsequent layer can begin concurrently with the computation of the current layer~\cite{li2024flexnn}. Moreover, in wireless edge networks, AI models can be partitioned into user-side and server-side sub-models for split inference by uploading intermediate features from end devices to edge servers~\cite{11301737}. To improve efficiency, this framework executes pipelined inference, allowing user-side inference to proceed in parallel with server-side inference on previously uploaded features. However, these methods focus on on-device inference and overlook the challenges of model downloading in wireless edge networks, where users share limited spectrum. In such settings, downlink rates are typically limited to hundreds of Mbps~\cite{cmhk5g}, far below the several to tens of GB/s available for on-device model loading~\cite{4090rate}. Thus, joint spectrum and computing resource allocation must be judiciously determined to support model downloading and inference in multi-user scenarios in wireless edge networks.

\section{The SLIDE Framework and Problem Formulation}
This section introduces the SLIDE framework, including the multi-user wireless edge network scenario and the procedure of SLIDE, formulates the E2E latency and energy consumption, and presents the task throughput maximization problem.
\subsection{Network Scenario}
As shown in Fig. \ref{fig:framework}, we consider a wireless network where multiple users $\mathcal{K}=\left\{1, 2, \dots, K \right\}$ connect to a BS with a cache hosting a model library $\mathcal{I}=\left\{1, 2, \dots, I\right\}$. We consider that each user requests an inference service with an E2E latency requirement $\bar{T}_{k}$ and can be served by a subset of models in $\mathcal{I}$, denoted by $\mathcal{I}_{k}$, where both $\bar{T}_{k}$ and $\mathcal{I}_{k}
$ are known to the BS. For instance, a user performing a classification task with a given accuracy requirement might be served by a number of model architectures, such as ResNet models or Transformer models with different bit widths. Without loss of generality, we only consider users who do not pre-store the models in $\mathcal{I}_{k}$ locally and should download a model from the BS via wireless channels to serve their needs. Let a binary variable $x_{k,i}$ indicate the model provisioning of the BS, where $x_{k,i}=1$ represents that the BS provisions model $i$ to user $k$, and 0 otherwise. The variable $x_{k,i}$ satisfies the constraint
\begin{equation}\label{const_1}
    \sum\limits_{i\in\mathcal{I}_{k}}x_{k,i}\le1, \ \forall k \in\mathcal{K},
\end{equation}
which ensures that user $k$ is provisioned at most one model from $\mathcal{I}_{k}$.

The frequently used notations are summarized in Table~\ref{table_notation}.

\begin{remark}
    Although a user can perform inference using the downloaded models multiple times, our optimization framework focuses on optimizing the \textit{first-time} model downloading when the user does not store the model locally. 
    As discussed in the Introduction, it is generally infeasible to pre-store or persistently cache all potentially requested models. Consequently, model cache misses are unavoidable in practice, and SLIDE serves as an enabling step to efficiently deliver the missing model from the BS by overlapping model downloading with inference, after which model caching, prefetching, and reuse can take place.
\end{remark}

\begin{figure}[t]
\centering
\includegraphics[width=0.4\textwidth]{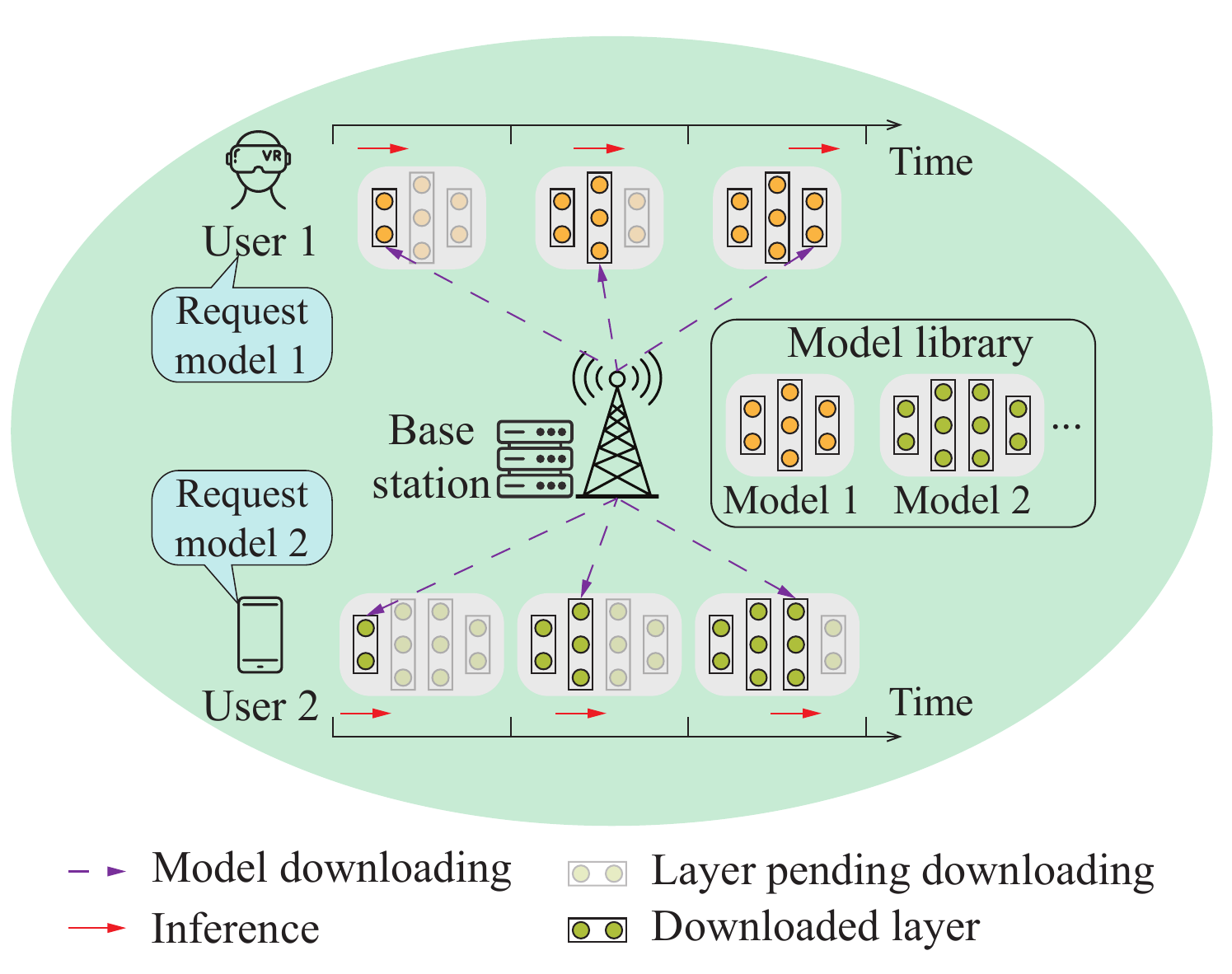}
\vspace{-0.25cm}
\caption{The proposed SLIDE framework, where users start inference with downloaded layers while simultaneously receiving the remaining layers.}
\vspace{-10pt}
\label{fig:framework}
\end{figure}

\begin{table}[t]
    \centering
	\caption{Frequently Used Notations}
	\label{table_notation}
	\begin{tabular}{|m{1.9cm}<{\centering}|m{6.1cm}<{\centering}|}
		\hline
		\textbf{Symbol} & \textbf{Description}\\ \hline
        $\mathcal{K}$ & Set of users. \\ \hline
        $\mathcal{I}$, $\mathcal{I}_{k}$ & Model library and set of models requested by user $k$. \\ \hline
        $K$, $I$, $L_{i}$ & Number of users, models, and layers of model $i$. \\ \hline
        $k$, $i$, $l_{i}$ & User index, model index, and layer index of model $i$. \\ \hline
        $x_{k,i}\in {\bf{X}}$ & Model provisioning indicator variable. \\ \hline
        $y_{k}\in{\bf{Y}}$ & Spectrum bandwidth allocation for user $k$. \\ \hline
        $z_{k,l_{i}}\in{\bf{Z}}$ & GPU frequency allocation scaling factor of user $k$ for layer $l_{i}$. \\ \hline
        ${\bf{Z}}_{k,i}$ & Vector of per-layer $z_{k,l_{i}}$ for user $k$ and model $i$.\\ \hline
        $S_{l_{i}}$ & Data size of layer $l_{i}$. \\ \hline
        $B$ & Total bandwidth of the BS. \\ \hline
        $R_{k}$, $\gamma_{k}$ & Spectral efficiency and channel gain between user $k$ and the BS. \\ \hline
        $P$, $N_{0}$ & Transmit power spectral density and noise power spectral density. \\ \hline
        $\bar{T}_{k}$, $Q_{k}$ & End-to-end latency requirement and maximum inference energy budget of user $k$. \\ \hline
        $\tau_{k,l_{i}}$ & Time for user $k$ to download layer $l_{i}$. \\ \hline
        $T_{k,i}\left(z_{k,l_{i}}\right)$ & Inference latency of user $k$ for layer $l_{i}$. \\ \hline
        $V_{1}\left(k,l_{i}\right)$ & Time required for user $k$ to transfer layer $l_{i}$ from system memory to GPU memory. \\ \hline
        $W_{l_{i}}$ & Computation workload of layer $l_{i}$ for processing one data sample. \\ \hline
        $b_{k}$, $\kappa_{k}$, $f_{k}$, $\Psi_{k}$ & Inference batch size, computing intensity, GPU clock frequency, and power coefficient of user $k$. \\ \hline
        $t_{k,l_{i}}\left(y_{k},{\bf{Z}}_{k,i}\right)$, $t_{k,L_{i}}\left(y_{k},{\bf{Z}}_{k,i}\right)$ & Time when user $k$ obtains the output of layer $l_{i}$ and the end-to-end latency of user $k$. \\ \hline
        $e_{k,i}\left({\bf{Z}}_{k,i}\right)$ & Inference energy consumption of user $k$ for model $i$. \\ \hline
        $e_{1}\left(k,i\right)$ & Energy consumption of user $k$ for model instantiation, data movement, and model transfer to GPU memory. \\ \hline
        $e_{k}\left({\bf{X}}_{k},{\bf{Z}}_{k}\right)$ & Total inference energy consumption of user $k$. \\ \hline
        $x^{*}_{k,i}\in{\bf{X}}^{*}$, $y^{*}_{k}\in{\bf{Y}}^{*}$, $z^{*}_{k,l_{i}}\in{\bf{Z}}^{*}$ & Optimal model provisioning, spectrum bandwidth allocation, and computing resource allocation. \\ \hline
        
	\end{tabular}
\end{table}

\subsection{The Procedure of SLIDE}
As shown in Fig.~\ref{fig:framework}, SLIDE overlaps model downloading and inference, allowing users to simultaneously perform inference with the downloaded layers while receiving the remaining layers of the model. This contrasts with conventional DAI (Fig.~\ref{fig:timeline}, top), where downloading and inference are separated.

\begin{figure}[t]
\centering
\includegraphics[width=0.48\textwidth]{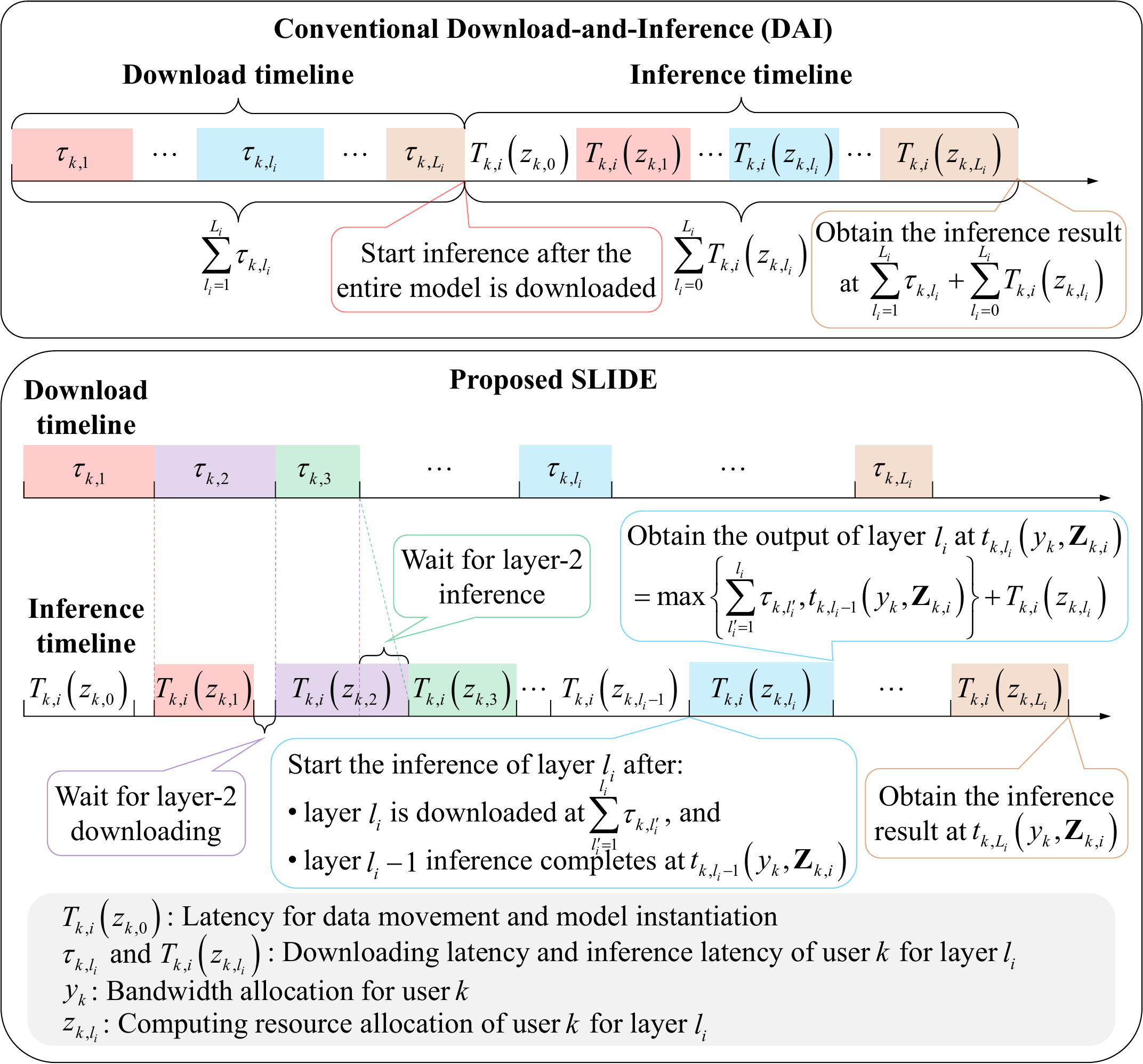}
\vspace{-0.25cm}
\caption{The procedures of conventional DAI and the proposed SLIDE framework, assuming that user $k$ is provisioned with model $i$. In conventional DAI, user $k$ begins inference only after downloading the entire model. In contrast, SLIDE enables user $k$ to start inference for each layer as soon as the layer is downloaded and the inference of the previous layer is completed.}
\vspace{-8pt}
\label{fig:timeline}
\end{figure}

The detailed procedure of SLIDE, as demonstrated in the lower part of Fig. \ref{fig:timeline}, is illustrated as follows. Suppose that user $k$ is provisioned with model $i$, with $\left\{1, 2, \dots, L_{i}\right\}$ being the ordered set of layers of model $i$. In SLIDE, a layer represents a model component (e.g., an individual model layer, a block, or a group of consecutive blocks), depending on the model architecture and chosen granularity. Starting from time 0, user $k$ continuously downloads the layers of model $i$ in the order of $l_{i}$ to the last layer $L_{i}$, while simultaneously performing layer-wise inference using the already downloaded layers. To further reduce the E2E latency, user $k$ moves task data to the GPU memory and executes model instantiation when downloading the first layer. Intuitively, the inference of user $k$ for layer $l_{i}$ begins at the latest of the following two events: (i) the completion of downloading layer $l_{i}$; and (ii) the completion of inference for the preceding layer $l_{i}-1$. 
As an exception, inference for layer~1 starts after the layer has been downloaded and the process of data movement and model instantiation is completed. Once the inference for layer $L_{i}$ is completed, user $k$ obtains the final inference result. Note that for model components whose internal computation is non-sequential (e.g., residual connections or parallel sub-module computation), SLIDE abstracts each such component as a single layer and incorporates its internal computation into the computation latency of that layer.

\begin{remark}
SLIDE supports feedforward models (e.g., convolutional neural networks (CNNs) like ResNets and Vision Transformers (ViTs)), autoregressive models (e.g., LLMs), and recurrent neural networks (RNNs). SLIDE enables parameter downloading for subsequent layers to overlap with the inference of the current layer. For autoregressive models, this overlap occurs during the prefilling stage, where input tokens are processed in a layer-wise feedforward manner before output generation. This is especially beneficial for latency reduction in long-context tasks, e.g., document summarization. For RNNs, this overlap applies at the first time step of inference. Moreover, while SLIDE focuses on full model downloads, it can extend to cases where users update models by downloading task-specific parameters (e.g., LoRA matrices for each LLM layer~\cite{hu2021lora}).
\end{remark}
\subsection{Latency Calculation}
Next, we provide the latency expression for SLIDE.
\subsubsection{Layer downloading latency}
We consider an orthogonal frequency-division multiple access (OFDMA) downlink transmission between the user and the BS. 
The time required for user $k$ to download layer $l_{i}$ alone is given by 
\begin{equation}\label{eq_tau}
    \tau_{k,l_{i}}=\frac{S_{l_{i}}}{y_{k}BR_{k}},
\end{equation}
where $S_{l_{i}}$ is the data size of layer $l_{i}$, and $B$ is the total bandwidth of the BS. $R_{k}={\rm{log}}_{2}\left(1+\frac{P\gamma_{k}}{N_{0}}\right)$, where $P$ is the transmit power spectral density of the BS, $\gamma_{k}$ is the channel gain between user $k$ and the BS, and $N_{0}$ is the spectral density of the additive white Gaussian noise. Moreover, $y_{k}\in\left[0,1\right]$ is the scaling factor of the bandwidth allocation\footnote{In OFDMA systems, channels are partitioned into time-frequency resource blocks (RBs). However, it is mathematically equivalent to consider bandwidth allocation for analysis, which can be mapped to RB allocation. For simplicity, we call it ``bandwidth allocation" in this paper.}, satisfying
\begin{equation}\label{const_2}
    \sum\limits_{k\in\mathcal{K}}y_{k}\le1.
\end{equation}

\subsubsection{Layer inference latency}
The inference latency of user $k$ with layer $l_{i}\in\left[1,L_{i}\right]$ is given by
\begin{equation}\label{eq_T}
    T_{k,i}\left(z_{k,l_{i}}\right)=V_{1}\left(k,l_{i}\right)+\frac{b_{k}W_{l_{i}}\kappa_{k}}{z_{k,l_{i}}f_{k}},
\end{equation}
where $V_{1}\left(k,l_{i}\right)$ denotes the time required for user $k$ to transfer layer $l_{i}$ from system memory to GPU memory\footnote{Layer- or block-wise parameter transfer and loading is supported by modern deep learning frameworks, e.g., PyTorch~\cite{pytorch_partial_Loading}. Without loss of generality, the GPU memory transfer time is represented by $V_{1}(k,l_i)$, typically approximated as the ratio of the data size of layer $l_i$ to the data rate between the system memory and GPU memory~\cite{xu2022igniter}.}. The second term in \eqref{eq_T} represents the time required for user $k$ to perform the forward propagation of layer $l_i$ on the GPU~\cite{lin2023efficient,zeng2021energy,yao2021evaluating}. Moreover, $b_{k}$ denotes the inference batch size of the input data samples of user $k$, $W_{l_{i}}$ represents the computation workload (in FLOPs) required by layer $l_{i}$ to process a single data sample, $\kappa_{k}$ denotes the computing intensity (in cycles/FLOP) of user $k$, and $f_{k}$ is the GPU clock frequency (in cycles/s) of user $k$. $z_{k,l_{i}}\in\left[0,1\right]$ is the GPU frequency allocation scaling factor of user $k$ for the forward propagation of layer $l_{i}$, satisfying 
\begin{equation}\label{const_3}
    z_{k,l_{i}}\le x_{k,i}\ \forall k \in\mathcal{K}, i\in\mathcal{I}_{k}, l_{i}\in\left[1,L_{i}\right],
\end{equation}
which ensures that computing resources are only allocated to the downloaded model.
Moreover, when $l_{i}=0$, $T_{k,i}\left(z_{k,0}\right)$ represents the latency of user $k$ for data movement and model instantiation, where $z_{k,0}=0$.


\subsubsection{The E2E latency}
By combining downloading and inference latency, we can obtain the E2E latency. First, as shown in Fig. \ref{fig:timeline}, the time at which user $k$ obtains the output of layer $l_{i}$ can be expressed as $t_{k,l_{i}}\left(y_{k},{\bf{Z}}_{k,i}\right)= \max\left\{\sum\limits_{l'_{i}=1}^{l_{i}}\frac{S_{l'_{i}}}{y_{k}BR_{k}},t_{k,l_{i}-1}\left(y_{k},{\bf{Z}}_{k,i}\right)\right\}+T_{k,i}\left(z_{k,l_{i}}\right)$, 
where ${\bf{Z}}_{k,i}$ is the vector of $z_{k,l_{i}}$, $\sum\limits_{l'_{i}=1}^{l_{i}}\frac{S_{l'_{i}}}{y_{k}BR_{k}}=\sum\limits_{l'_{i}=1}^{l_{i}}\tau_{k,l'_{i}}$ is the time for user $k$ to download the first $l_i$ layers of model $i$, $t_{k,l_{i}-1}\left(y_{k},{\bf{Z}}_{k,i}\right)$ denotes the time at which user $k$ obtains the output of layer $l_{i}-1$, and $t_{k,0}\left(y_{k},{\bf{Z}}_{k,i}\right)=T_{k,i}\left(z_{k,0}\right)$. 


Next, the E2E latency of user $k$, when provisioned with model $i$, is the time for user $k$ to obtain the output of the last layer of model $i$, which is hence given by
\begin{equation}\label{eq_final_e2e}
    \begin{aligned}
        &t_{k,L_{i}}\left(y_{k},{\bf{Z}}_{k,i}\right)\\ 
        &= \max\left\{\sum\limits_{l_{i}=1}^{L_{i}}\frac{S_{l_{i}}}{y_{k}BR_{k}},t_{k,L_{i}-1}\left(y_{k},{\bf{Z}}_{k,i}\right)\right\}+T_{k,i}\left(z_{k,L_{i}}\right),
    \end{aligned}
\end{equation}
which exhibits a recursive dependence across layers and is influenced by layer-wise downloading and inference latencies.

\subsection{Energy Consumption on Users}
Given that the process only involves downlink transmissions and the receiving power is typically small, we only consider the inference energy consumption on users. The inference energy consumption of user $k$ with model $i$ is given by
\begin{equation}\label{eq_e_k_i}
    e_{k,i}\left({\bf{Z}}_{k,i}\right)=e_{1}\left(k,i\right)+\sum\limits_{l_{i}=1}^{L_{i}}\Psi_{k}\kappa^{3}_{k} b_{k}W_{l_{i}}\left(\frac{z_{k,l_{i}}f_{k}}{\kappa_{k}}\right)^{2},
\end{equation}
where $e_{1}\left(k,i\right)$ denotes the energy consumption of user $k$ for the process of instantiating model $i$, moving input data, and transferring model $i$ to the GPU memory\footnote{For generality, we use a constant $e_{1}\left(k,i\right)$ to represent the energy consumption of these operations.}. The second term in \eqref{eq_e_k_i} denotes the energy consumption of user $k$ for forward propagation using model $i$ \cite{zeng2021energy,liu2012power,yao2021evaluating}, 
where $\Psi_{k}$ is the power coefficient (in Watt/$\text{(cycle/s)}^3$). Furthermore, the total inference energy consumption of user $k$ is expressed as 
\begin{equation}
    e_{k}\left({\bf{X}}_{k},{\bf{Z}}_{k}\right)=\sum\limits_{i\in\mathcal{I}_{k}}x_{k,i}e_{k,i}\left({\bf{Z}}_{k,i}\right),
\end{equation}
where ${\bf{X}}_{k}$ is the vector of $x_{k,i}$. Moreover, 
$e_{k}\left({\bf{X}}_{k},{\bf{Z}}_{k}\right)$ is subject to
\begin{equation}\label{const_4}
    e_{k}\left({\bf{X}}_{k},{\bf{Z}}_{k}\right)\le Q_{k},\ \forall k \in\mathcal{K},
\end{equation}
where $Q_{k}$ is the maximum inference energy budget of user $k$.
\subsection{Task Throughput Maximization Problem}
By considering latency and energy constraints, SLIDE aims to maximize the number of served users by jointly optimizing model provisioning $x_{k,i} \in {\bf{X}}$, spectrum bandwidth allocation $y_{k} \in {\bf{Y}}$, and computing resource allocation $z_{k,l_{i}} \in {\bf{Z}}$. The problem formulation is given as follows.
\begin{subequations}
	\begin{equation}
		{\mathcal{P}1}:\mathop{\max}\limits_{{\bf{X}},{\bf{Y}},{\bf{Z}}}\ U\left({\bf{X}},{\bf{Y}},{\bf{Z}}\right)=\sum\limits_{k\in\mathcal{K}}\sum\limits_{i\in\mathcal{I}}x_{k,i}\mathbb{I}_{\left\{t_{k,L_{i}}\left(y_{k},{\bf{Z}}_{k,i}\right)\le \bar{T}_{k}\right\}}
	\end{equation}	
        \begin{equation}
		{\rm{s.t.}} \ \eqref{const_1},\eqref{const_2},\eqref{const_3},\eqref{const_4},
	\end{equation}	
        \begin{equation}\label{const_8}
		x_{k,i}\le \mathbb{I}_{\left\{t_{k,L_{i}}\left(y_{k},{\bf{Z}}_{k,i}\right)\le \bar{T}_{k}\right\}},\ \forall k \in\mathcal{K},i\in\mathcal{I}_{k}, 
	\end{equation}	
        \begin{equation}\label{const_5}
		x_{k,i}\in\left\{0,1\right\},\ \forall k \in\mathcal{K}, i\in\mathcal{I}_{k},
	\end{equation}	
	\begin{equation}\label{const_6}
		y_{k}\in\left[0,1\right],\ \forall k \in\mathcal{K},
	\end{equation}	
        \begin{equation}\label{const_7}
		z_{k,l_{i}}\in\left[0,1\right],\ \forall k \in\mathcal{K}, l_{i}\in\left[1, L_{i}\right].
	\end{equation}	
\end{subequations}
Here, constraint \eqref{const_1} guarantees that user $k$ downloads at most one model from the BS, and constraint \eqref{const_2} limits the total bandwidth usage at the BS. Constraint \eqref{const_3} enforces that user $k$ allocates computing resources only to the downloaded model. Constraint \eqref{const_4} ensures user $k$'s energy consumption is within the budget. Constraint \eqref{const_8} ensures model $i$ can be provisioned to user $k$ if user $k$ can obtain the inference result using model $i$ within $\bar{T}_k$, where the binary indicator function $\mathbb{I}_{\left\{t_{k,L_{i}}\left(y_{k},{\bf{Z}}_{k,i}\right)\le \bar{T}_{k}\right\}}=1$ if and only if $t_{k,L_{i}}\left(y_{k},{\bf{Z}}_{k,i}\right)\le \bar{T}_{k}$.
\begin{remark}
    $z_{k,l_{i}}$ determines the per-layer GPU frequency of user $k$ for layer $l_{i}$ by considering the varying layer sizes and workloads. To avoid idle time between the completion of inference of the current layer and that of downloading the next layer, per-layer computing resource allocation and spectrum bandwidth allocation must be jointly optimized to maximize system capacity under latency and energy budgets.
\end{remark}

Note that $\mathcal{P}1$ is formulated using a static ``snapshot" of user locations and ignores user mobility, similar to prior work~\cite{9843917,9619857}. However, our framework can be easily extended to mobile scenarios by considering the worst-case channel gain within the latency deadline, which can be estimated by predicting user trajectories and corresponding channel conditions~\cite{9007469,7981536}, when calculating \eqref{eq_tau}. We will provide simulation results in Section V-C to demonstrate that our algorithm is robust under user mobility. 

In addition, although $\mathcal{P}1$ is formulated for a single decision time slot, long-term user fairness across multiple time slots can be incorporated without changing the core procedure of SLIDE. This can be achieved by introducing time-varying user weights that increase for previously unserved users~\cite{7976338}, and by sequentially solving the per-slot optimization problem. Since the user weights are known at the beginning of each decision time slot, the per-slot optimization problem preserves the same problem property as $\mathcal{P}1$, without affecting the subsequent algorithm design.

$\mathcal{P}1$ is a mixed-integer nonlinear programming (MINLP) problem involving both integer and continuous decision variables, making it challenging to solve efficiently. Conventional methods for MINLP problems, such as the Branch-and-Bound algorithm, have exponential time complexity for obtaining optimal solutions. To address this issue, we propose an efficient algorithm to solve $\mathcal{P}1$, detailed in the next section.

\section{Algorithm Design}
This section develops a polynomial-time optimal solution approach to $\mathcal{P}1$. We first determine the model provisioning and computing resource allocation for each user given the bandwidth allocation. Based on this, we \textit{equivalently} transform $\mathcal{P}1$ into a bandwidth allocation problem $\mathcal{P}3$, followed by the determination of model provisioning and computing resource allocation. To solve $\mathcal{P}3$, we design an iterative algorithm to determine each user's minimum feasible bandwidth allocation, and then propose an algorithm that efficiently yields the optimal solution to $\mathcal{P}1$.
\subsection{Model Provisioning and Computing Resource Allocation}
\subsubsection{Computing resource allocation for user $k$ with model $i$ given $y_{k}$}
To conquer Problem $\mathcal{P}1$, we begin by deriving the computing resource allocation for user $k$ with model $i\in\mathcal{I}_{k}$, assuming that the spectrum bandwidth allocation is given for user $k$. The details are summarized in the following proposition. 
\begin{proposition}\label{proposition_1}
    Given spectrum bandwidth allocation $y_{k}\in{\bf{Y}}$, the optimal computing resource allocation for user $k$ with model $i$, denoted by $\hat{z}^{*}_{k,l_{i}}\in\hat{{\bf{Z}}}^{*}_{k,i}$, can be obtained by solving the following problem $\mathcal{P}2$, where $\hat{z}_{k,l_{i}}\in\hat{{\bf{Z}}}_{k,i}$.
\end{proposition}
\vspace{-15pt}
\begin{subequations}
	\begin{equation}
		{\mathcal{P}2}:\ t_{k,L_{i}}\left(y_{k},\hat{{\bf{Z}}}^{*}_{k,i}\right)=\mathop{\min}\limits_{\hat{{\bf{Z}}}_{k,i}}\ t_{k,L_{i}}\left(y_{k},\hat{{\bf{Z}}}_{k,i}\right)
	\end{equation}	
        \begin{equation}\label{const_p2_1}
		{\rm{s.t.}} \ e_{1}\left(k,i\right)+\sum\limits_{l_{i}=1}^{L_{i}}\Psi_{k}\kappa^{3}_{k} b_{k}W_{l_{i}}\left(\frac{\hat{z}_{k,l_{i}}f_{k}}{\kappa_{k}}\right)^{2}\le Q_{k},
	\end{equation}	
        \begin{equation}
            \hat{z}_{k,l_{i}}\in\left[0,1\right],\ \forall k \in\mathcal{K}, i\in\mathcal{I}_{k},
        \end{equation}
        \begin{equation}\label{const_p2_3}
            \hat{z}_{k,l_{i}}\le \mathbb{I}_{\left\{y_{k}>0\right\}},\ \forall k \in\mathcal{K}, i\in\mathcal{I}_{k},
        \end{equation}
    \end{subequations}   
where \eqref{const_p2_3} ensures that no computing resources of user $k$ are allocated to model $i$ when bandwidth allocation $y_{k}=0$, where the binary indicator function $\mathbb{I}_{\left\{y_{k}>0\right\}}=1$ if and only if $y_{k}>0$.
\begin{proof}
    The proof is provided in Appendix \myref{proof_proposition_1}.
\end{proof}

Next, we elaborate on how to solve $\mathcal{P}2$.
%
We begin by deriving the following proposition for $\mathcal{P}2$, where $e_{k,i}\left({\bf{1}}\right)$ denotes the energy consumption of user $k$ to perform forward propagation using model $i$ with $\hat{z}_{k,l_{i}}=1$ for all layers.


\begin{proposition}\label{lemma_0}
    For the optimal solution $\hat{{\bf{Z}}}^{*}_{k,i}$ to $\mathcal{P}2$, $t_{k,l_{i}-1}\left(y_{k},\hat{{\bf{Z}}}^{*}_{k,i}\right)\ge \sum\limits_{l'_{i}=1}^{l_{i}}\tau_{k,l'_{i}}$ holds when $y_{k}\ne 0$ and $e_{k,i}\left({\bf{1}}\right)\le Q_{k}$, $\forall l_{i}\in\left[2,L_{i}\right]$.
\end{proposition}
\begin{proof}
    The proof is shown in Appendix \myref{proof_lemma_0}.
\end{proof}

\textbf{Takeaway for Proposition \ref{lemma_0}}: Proposition \ref{lemma_0} shows that, given $y_{k}\ne 0$ and $e_{k,i}\left({\bf{1}}\right)\le Q_{k}$, the inference completion time of layer $l_{i}-1$ is no earlier than the downloading completion time of layer $l_{i}$, $\forall l_{i}\in\left[2,L_{i}\right]$. In other words, the above conditions indicate that inference of layer $l_{i}$ starts right after inference completion of layer $l_{i}-1$, ensuring continuous forward propagation without idle time between successive layers. 

With Proposition \ref{lemma_0}, we have the following result for $\hat{z}^{*}_{k,l_{i}}$.
\begin{figure*}[b]
\begin{equation}\label{eq_update}
    \begin{aligned}
        &
        \rho^{\left(m\right)}_{l_{i}}=1-\sum\limits_{l'_{i}=l_{i}}^{L_{i}}\mu^{\left(m\right)}_{l'_{i}},\quad
        \eta^{\left(m\right)}=\chi\left(\mu^{\left(m\right)}_{l_{i}},Q'_{k}\right),\quad
        \hat{z}^{\left(m\right)}_{k,l_{i}}= \min\left\{1,\max\left\{0,\sqrt[3]{\frac{\rho^{\left(m\right)}_{l_{i}}}{2\eta^{\left(m\right)}}}\right\}\right\},\\
        &\mu^{\left(m+1\right)}_{l_{i}}
        =\begin{cases}
                    \max\left\{0,\mu^{\left(m\right)}_{l_{i}}+\delta^{\left(m\right)}_{\mu_{l_{i}}}\left[ \sum\limits_{l_{i}'=1}^{l_{i}+1}\tau_{k,l'_{i}}-s_{k,1}-\sum\limits_{l'_{i}=1}^{l_{i}}V_{1}\left(k,l'_{i}\right)-\sum\limits_{l'_{i}=1}^{l_{i}}\frac{\Gamma_{k,l'_{i}}}{\tilde{z}^{\left(m\right)}_{k,l'_{i}}}\right]\right\},\ \text{if }1\le l_{i}\le L_{i}-1,\\
				0,\ \text{if }l_{i}=L_{i}.
        \end{cases}\\
    \end{aligned}
    \end{equation}
\vspace{-10pt}
\end{figure*}
\begin{proposition}\label{lemma_2}
    The optimal solution $\hat{z}^{*}_{k,l_{i}}$ to $\mathcal{P}2$ is given by
    \begin{equation}\label{eq_z_2}
        \hat{z}^{*}_{k,l_{i}}=\\
        \begin{cases}
            0, \ \text{if }y_{k}=0,\\
            1, \ \text{if }y_{k}\ne0 \text{ and }e_{k,i}\left({\bf{1}}\right)\le Q_{k},\\
            \begin{aligned}
            \dot{z}_{k,l_{i}}, 
            &\text{if }y_{k}\ne0 ,e_{k,i}\left({\bf{1}}\right)> Q_{k}, e_{k,i}\left(\dot{{\bf{Z}}}_{k,i}\right)\le Q_{k}, \\
            &\text{ and }{\bf{0}}\preceq\dot{{\bf{Z}}}_{k,i} \preceq{\bf{1}},
            \end{aligned}
            \\
            \min\left\{1,\max\left\{0,\sqrt[3]{\frac{\rho^{*}_{l_{i}}}{2\eta^{*}}}\right\}\right\}, \ \text{otherwise}.\\
        \end{cases}
    \end{equation}
    Here, $\dot{z}_{k,l_{i}}\in \dot{{\bf{Z}}}_{k,i}$ is defined in \eqref{eq_dot_z}, where $\Gamma_{k,l_{i}}=\frac{b_{k}W_{l_{i}}\kappa_{k}}{f_{k}}$ and $Q'_{k}=\frac{Q_{k}-e_{1}\left(k,i\right)}{\Psi_{k}f^{3}_{k}}$. Moreover, $\rho^{*}_{l_{i}}=1-\sum\limits_{l'_{i}=l_{i}}^{L_{i}}\mu^{*}_{l'_{i}}$, and $\mu^{*}_{l_{i}}$ and $\eta^{*}$ denote the optimal Lagrange multipliers, which can be determined via iteration using the update rules in \eqref{eq_update} \cite{1664999}, with $\mu^{*}_{L_{i}}=0$. In \eqref{eq_update}, $\mu^{\left(m\right)}_{l_{i}}$, $\eta^{\left(m\right)}$,  $\tilde{z}^{\left(m\right)}_{k,l_{i}}$, and $\rho^{\left(m\right)}_{l_{i}}$ represent the values of $\mu_{l_{i}}$, $\eta$, $\tilde{z}_{k,l_{i}}$, and $\rho_{l_{i}}$ in the $m$-th iteration, respectively, and $\delta^{\left(m\right)}_{\mu_{l_{i}}}$ is the step size for updating $\mu^{\left(m\right)}_{l_{i}}$.
    Additionally, $\chi\left(\mu^{\left(m\right)}_{l_{i}},Q'_{k}\right)$ is the inverse function that determines $\eta^{\left(m\right)}$ from $\sum\limits_{l_{i}=1}^{L_{i}}\Gamma_{k,l_{i}}\min\left\{1,\max\left\{0,\left(\frac{\rho^{\left(m\right)}_{l_{i}}}{2\eta^{\left(m\right)}}\right)^{\frac{2}{3}}\right\}\right\}=Q'_{k}$ in each iteration.
\end{proposition}
\begin{equation}\label{eq_dot_z}
    \dot{z}_{k,l_{i}}=
    \begin{cases}
        \frac{\Gamma_{k,1}}{\tau_{k,1}+\tau_{k,2}-s_{k,1}-V_{1}\left(k,1\right)}, \text{ if } l_{i}=1,\\
        \frac{\Gamma_{k,l_{i}}}{\tau_{k,l_{i}+1}-V_{1}\left(k,l_{i}\right)},\text{ if } 2\le l_{i}\le L_{i}-1,\\
        1, \text{ if } l_{i}=L_{i}.
    \end{cases}
\end{equation}
\begin{proof}
    The proof is shown in Appendix \myref{proof_lemma_2}.
\end{proof}
Based on Proposition \ref{lemma_2}, we have the following corollary for the computing resource allocation and the energy consumption in the optimal solution to $\mathcal{P}1$.
\begin{corollary}\label{remark_2}
    For the optimal solution to $\mathcal{P}1$, if user $k$ is provisioned with model $i$ and assigned a positive $y_{k}$, then the following results hold.
    \begin{itemize}
        \item If $e_{k,i}\left({\bf{1}}\right)\le Q_{k}$, then user $k$ should perform forward propagation using model $i$ with the maximum GPU clock frequency.
        \item If both $e_{k,i}\left({\bf{1}}\right)$ and $e_{k,i}\left(\dot{{\bf{Z}}}_{k,i}\right)$ exceed $Q_{k}$, user $k$ will consume the entire energy budget to complete forward propagation with model $i$.
    \end{itemize}
\end{corollary}

\begin{proof}
    The proof is shown in Appendix \myref{proof_remark_2}.    
\end{proof}
\subsubsection{Model provisioning and computing resource allocation for user $k$ given $y_{k}$}
Based on $\hat{z}^{*}_{k,l_{i}}$ derived in Proposition~\ref{lemma_2}, we establish the following proposition.
\begin{proposition}\label{theorem_1}
    Given spectrum bandwidth allocation $y_{k}\in{\bf{Y}}$, the optimal model provisioning for user $k$ is 
    \begin{equation}\label{eq_x_optimal}
        \tilde{x}^{*}_{k,i}
        =\left\{ {\begin{array}{*{20}{c}}
                    \begin{aligned}
				&1,\ \text{if }y_{k}\ne 0\text{ and }i=i^{*} ,\\
                    &0,\ \text{otherwise},
                    \end{aligned}
		\end{array}} \right.
    \end{equation}
    and the optimal computing resource allocation is 
    \begin{equation}\label{eq_z_optimal}
        \tilde{z}^{*}_{k,l_{i}}
        =\left\{ {\begin{array}{*{20}{c}}
                    \begin{aligned}
				&\hat{z}^{*}_{k,l_{i}},\ \text{if }y_{k}\ne 0 \text{ and } i=i^{*},\\
                    &0,\ \text{otherwise},
                    \end{aligned}
		\end{array}} \right.
    \end{equation}
    where $i^{*}=\mathop{\arg\min}\limits_{i\in\mathcal{I}_{k}}\left\{t_{k,L_{i}}\left(y_{k},\hat{{\bf{Z}}}^{*}_{k,i}\right)\  \middle| \ t_{k,L_{i}}\left(y_{k},\hat{{\bf{Z}}}^{*}_{k,i}\right)\le \bar{T}_{k}\right\}$,
    and $t_{k,L_{i}}\left(y_{k},\hat{{\bf{Z}}}^{*}_{k,i}\right)$ is the minimum E2E latency of user $k$ with model $i$ under $y_{k}$, which can be obtained by solving $\mathcal{P}2$. 
    
\end{proposition}
\begin{proof}
    On the one hand, if $y_{k}= 0$ or no model in $\mathcal{I}_{k}$ satisfies the latency constraint, then no feasible model provisioning or computing resource allocation exists that enables user $k$ to satisfy $\mathbb{I}_{\left\{t_{k,L_{i}}\left(y_{k},{\bf{Z}}_{k,i}\right)\le \bar{T}_{k}\right\}}=1$. In this case, both $\tilde{x}^{*}_{k,i}$ and $\tilde{z}^{*}_{k,l_{i}}$ are set to 0. On the other hand, if $y_{k}\ne 0$ and there exists at least one model in $\mathcal{I}_{k}$ satisfying the latency constraint, then provisioning user $k$ with model $i^{*}$ and setting $\tilde{z}^{*}_{k,l_{i^{*}}}=\hat{z}^{*}_{k,l_{i^{*}}}$ achieves the minimum E2E latency, thus preserving the optimality. This completes the proof.
\end{proof}

\subsection{Equivalent Problem Transformation}
To simplify solving $\mathcal{P}1$, we leverage Proposition~\ref{lemma_2} to decouple both model provisioning and computing resource allocation from the original problem, leading to the following proposition.
\begin{proposition}\label{theorem_2}
    Solving $\mathcal{P}1$ is equivalent to first solving the following problem $\mathcal{P}3$ on bandwidth allocation and then determining the model provisioning and computing resource allocation using \eqref{eq_x_optimal} and \eqref{eq_z_optimal}, respectively.
\end{proposition}
\begin{subequations}
	\begin{equation}
		{\mathcal{P}3}:\ \mathop{\max}\limits_{{\bf{Y}}}\ U\left({\tilde{\bf{X}}}^{*},{\bf{Y}},\tilde{{\bf{Z}}}^{*}\right)
	\end{equation}	
        \begin{equation}
		{\rm{s.t.}} \ \eqref{const_2},\eqref{const_8},\eqref{const_6},
	\end{equation}	
\end{subequations}
where $\tilde{\bf{X}}^{*}$ and ${\tilde{{\bf{Z}}}}^{*}$ are the vectors of $\tilde{x}^{*}_{k,i}$ and $\tilde{z}^{*}_{k,l_{i}}$, respectively. They are functions of ${\bf{Y}}$, which can be obtained from \eqref{eq_x_optimal} and \eqref{eq_z_optimal}, respectively.
\begin{proof}
    From Proposition \ref{theorem_1}, the optimal model provisioning and computing resource allocation for user $k$, under any given $y_{k}$, can be derived from \eqref{eq_x_optimal} and \eqref{eq_z_optimal}, respectively. Substituting $t_{k,L_{i}}\left(y_{k},\tilde{\bf{Z}}_{k,i}\right)$ into $\mathcal{P}1$ and removing the constraints on $\bf{X}$ and $\bf{Z}$ lead to an equivalent problem $\mathcal{P}3$, which depends solely on $\bf{Y}$. Therefore, solving $\mathcal{P}3$ for the spectrum bandwidth allocation, followed by determining the model provisioning and computing resource allocation based on \eqref{eq_x_optimal} and \eqref{eq_z_optimal}, respectively, obtains the optimal solution to $\mathcal{P}1$, which completes the proof.
\end{proof}
\subsection{Optimal Solution Approach}
Based on Proposition \ref{theorem_2}, this subsection derives the optimal solution to $\mathcal{P}1$. We begin by determining the minimum feasible bandwidth allocation $\check{y}_k$ for user $k$, as defined in Definition~\ref{definition_1}. Based on $\check{y}_k$, we then derive the optimal bandwidth allocation for $\mathcal{P}3$ and obtain the optimal solution to $\mathcal{P}1$.
\begin{definition}\label{definition_1}
     \textbf{Minimum feasible bandwidth allocation}: Given a constant $\check{y}_{k}\in\left[0,1\right]$, the minimum feasible bandwidth allocation for user $k$ is the smallest scaling factor of its bandwidth allocation such that no feasible model provisioning and computing resource allocation can enable user $k$ to complete its inference task within $\bar{T}_{k}$ if the allocated bandwidth is less than $\check{y}_{k}$.
\end{definition}

\subsubsection{Minimum feasible bandwidth allocation for a user}\label{sec_minum}
We propose an iterative algorithm to derive $\check{y}_{k}$ along with the corresponding model provisioning $\check{x}^{*}_{k,i}$ and computing resource allocation $\check{z}^{*}_{k,l_{i}}$ for user $k$, as detailed in Algorithm~\ref{algorithm_bi}. The algorithm iteratively searches $\check{y}_{k}$ within the range $\left[\check{y}^{\left(\min\right)}_{k}, \check{y}^{\left(\max\right)}_{k}\right]$ until the desired error bound $\epsilon$ is satisfied. The initial values of $\check{y}^{\left(\min\right)}_{k}$ and $\check{y}^{\left(\max\right)}_{k}$ are set to $\frac{\min\limits_{i\in\mathcal{I}_{k}}\sum\limits_{l_{i}=1}^{L_{i}}S_{l_{i}}}{\bar{T}_{k}BR_{k}}$ and 1, respectively, in Line~\ref{line:bi_initialization}, where $\min\limits_{i\in\mathcal{I}_{k}}\sum\limits_{l_{i}=1}^{L_{i}}S_{l_{i}}$ is the minimum model size among models in $\mathcal{I}_{k}$. Then, Algorithm~\ref{algorithm_bi} adopts the bisection search method to determine $\check{y}_{k}$ and the corresponding $\check{x}^{*}_{k,i}$ and $\check{z}^{*}_{k,l_{i}}$. 
Specifically, in each iteration of the while loop from Line~\ref{line:bi_while_start} to \ref{line:bi_while_end}, $\check{y}_{k}$ is first updated to the midpoint of the interval $\left[\check{y}^{\left(\min\right)}_{k}, \check{y}^{\left(\max\right)}_{k}\right]$. Next, from Line~\ref{line:bi_for_start} to \ref{line:bi_for_end}, Algorithm~\ref{algorithm_bi} calculates $\hat{z}^{*}_{k,l_{i}}$ for each model in $\mathcal{I}_{k}$ under current $\check{y}_{k}$ based on \eqref{eq_z_2}, and obtains the corresponding minimum E2E latency $t_{k,L_{i}}\left(\check{y}_{k},\hat{{\bf{Z}}}_{k,i}^{*}\right)$. Then, in Line \ref{line:bi_i_*}, the algorithm identifies model $i^{*}$ for user $k$ with the minimum $t_{k,L_{i}}\left(\check{y}_{k},\hat{{\bf{Z}}}_{k,i}^{*}\right)$ among the models satisfying the latency requirement $\bar{T}_{k}$. If such a model $i^{*}$ exists, then $\check{y}^{\left(\max\right)}_{k}$ is decreased to $\check{y}_{k}$; otherwise, $\check{y}^{\left(\min\right)}_{k}$ is increased to $\check{y}_{k}$. Upon convergence, if $i^{*}$ has been identified in Line~\ref{line:bi_i_*}, then the algorithm determines $\check{y}_{k}$ in Line~\ref{line:bi_output_y}, calculates $\tilde{x}^{*}_{k,i}$ and $\tilde{z}^{*}_{k,l_{i}}$ in Line \ref{line:bi_tilde_x_z}, and obtains $\check{x}^{*}_{k,i}$ and $\check{z}^{*}_{k,l_{i}}$ corresponding to $\check{y}_{k}$ in Line~\ref{line:bi_check_z}. Otherwise, $\check{y}_{k}$, $\check{x}^{*}_{k,i}$, and $\check{z}^{*}_{k,l_{i}}$ are set to 0.
\begin{algorithm}[!t]
	\caption{Minimum Feasible Bandwidth Allocation Algorithm} 
	\label{algorithm_bi}
	\LinesNumbered
	\KwIn{$k$ and $\epsilon$.}
	\KwOut{$\check{y}_{k}$, $\check{x}^{*}_{k,i}$, and $\check{z}^{*}_{k,l_{i}}$.} 
            {\bf Initialize:} $\check{y}_{k}=0$, $\check{x}^{*}_{k,i}=0$, and $\check{z}^{*}_{k,l_{i}}=0$. $\check{y}^{\left(\min\right)}_{k}=\frac{\min\limits_{i\in\mathcal{I}_{k}}\sum\limits_{l_{i}=1}^{L_{i}}S_{l_{i}}}{\bar{T}_{k}BR_{k}}$ and $\check{y}^{\left(\max\right)}_{k}=1$.\label{line:bi_initialization}\\
            \If{$\check{y}^{\left(\min\right)}_{k}>1$}
            {
                \textbf{Return}.\\
            }
            \While{$\left|\check{y}^{\left(\min\right)}_{k}-\check{y}^{\left(\max\right)}_{k}\right|\ge \epsilon$}
            {\label{line:bi_while_start}
                $\check{y}_{k}=\frac{\check{y}^{\left(\min\right)}_{k}+\check{y}^{\left(\max\right)}_{k}}{2}$.\\
                \For{$i\in\mathcal{I}_{k}$}
                {\label{line:bi_for_start}
                    Calculate $\hat{z}^{*}_{k,l_{i}}$ from \eqref{eq_z_2} by solving $\mathcal{P}2$ under $\check{y}_{k}$ and obtain the corresponding $t_{k,L_{i}}\left(\check{y}_{k},\hat{{\bf{Z}}}_{k,i}^{*}\right)$.\label{line:bi_solve_p2}\\
                }\label{line:bi_for_end}
                $i^{*}=\mathop{\arg\min}\limits_{i \in \mathcal{I}_{k}}\left\{t_{k,L_{i}}\left(\check{y}_{k},\hat{{\bf{Z}}}_{k,i}^{*}\right)\  \middle| \ t_{k,L_{i}}\left(\check{y}_{k},\hat{{\bf{Z}}}_{k,i}^{*}\right)\le \bar{T}_{k}\right\}$.\label{line:bi_i_*}\\
                \eIf{$i^{*}$ exists}
                {\label{line:bi_if_start}
                    $\check{y}^{\left(\max\right)}_{k}=\check{y}_{k}$.\\
                }
                {
                    $\check{y}^{\left(\min\right)}_{k}=\check{y}_{k}$.\\
                }\label{line:bi_if_end}
            }\label{line:bi_while_end}
            \eIf{$i^{*}$ exists}
            {\label{line:bi_if_i_start}
                $\check{y}_{k}=\check{y}^{\left(\max\right)}_{k}$. \label{line:bi_output_y}\\
                Calculate $\tilde{x}^{*}_{k,i}$ and $\tilde{z}^{*}_{k,l_{i}}$ under $\check{y}_{k}$ from \eqref{eq_x_optimal} and \eqref{eq_z_optimal}, respectively. \label{line:bi_tilde_x_z}\\
                $\check{x}^{*}_{k,i}=\tilde{x}^{*}_{k,i}$, and $\check{z}^{*}_{k,l_{i}}=\tilde{z}^{*}_{k,l_{i}}$. \label{line:bi_check_z}\\
            }
            {
                $\check{y}_{k}=0$, $\check{x}^{*}_{k,i}=0$, and $\check{z}^{*}_{k,l_{i}}=0$.\\
            }\label{line:bi_if_i_end}
\end{algorithm}

\subsubsection{Optimal solution to $\mathcal{P}1$}
\begin{algorithm}[!t]
	\caption{Task Throughput Maximization Algorithm} 
	\label{algorithm_greedy}
	\LinesNumbered
	\KwIn{$\mathcal{K}$ and $\mathcal{I}_{k}$.}
	\KwOut{${\bf{X}}^{*}$, ${\bf{Y}}^{*}$, ${\bf{Z}}^{*}$, and $U^{*}$.} 
	{\bf Initialize:} $\mathcal{K}'=\mathcal{K}$ and $U^{*}=0$. Set $\mathcal{K}^{*}=\emptyset$, and set ${\bf{X}}^{*}$, ${\bf{Y}}^{*}$, and ${\bf{Z}}^{*}$ as ${\bf{0}}$.\\
        Calculate $\check{y}_{k}$, $\check{x}^{*}_{k,i}$, and $\check{z}^{*}_{k,l_{i}}$ for $k\in\mathcal{K}$ from Algorithm~\ref{algorithm_bi}. \label{line:greedy_call}\\
        \While{$\sum\limits_{k\in\mathcal{K}^{*}}y^{*}_{k}\le 1$ and $\exists k$, $\check{y}_{k}>0$}
        {\label{line:greedy_while_start}
            $k^{*}=\mathop{\arg\min}\limits_{k \in \mathcal{K}'}\left\{\check{y}_{k}\right\}$. \label{line:greedy_k}\\
            \If{$\sum\limits_{k\in\mathcal{K}^{*}}y^{*}_{k} + y^{*}_{k^{*}}> 1$}
            {
            \textbf{Break}.\\
            }
            
            $y^{*}_{k^{*}}=\check{y}_{k^{*}}$, $x^{*}_{k^{*},i}=\check{x}^{*}_{k^{*},i}$, and $z^{*}_{k^{*},l_{i}}=\check{z}^{*}_{k^{*},l_{i}}$. \label{line:greedy_z}\\ 
            $U^{*}=U^{*}+1$. $\mathcal{K}^{*}=\mathcal{K}^{*}\cup \left\{k^{*}\right\}$. $\mathcal{K}'=\mathcal{K}'\setminus \left\{k^{*}\right\}$.\label{line:greedy_u}\\
        }     \label{line:greedy_while_end}   
\end{algorithm}
Based on $\check{y}_{k}$ calculated from Algorithm \ref{algorithm_bi}, we propose an efficient algorithm to obtain the maximum task throughput $U^{*}$ and determine the optimal $x^{*}_{k,i}\in{\bf{X}}^{*}$, $y^{*}_{k}\in{\bf{Y}}^{*}$, and $z^{*}_{k,l_{i}}\in{\bf{Z}}^{*}$ to $\mathcal{P}1$. The proposed algorithm is outlined in Algorithm \ref{algorithm_greedy} and detailed below. Algorithm \ref{algorithm_greedy} begins by invoking Algorithm \ref{algorithm_bi} in Line \ref{line:greedy_call} to calculate $\check{y}_{k}$, $\check{x}^{*}_{k,i}$, and $\check{z}^{*}_{k,l_{i}}$ for all users. Then, in each iteration of the while loop from Line~\ref{line:greedy_while_start} to \ref{line:greedy_while_end}, the algorithm first identifies user $k^{*}$ with the minimum value of $\check{y}_{k}$ among the users in $\mathcal{K}'$, which is the set of users yet to be served, in Line \ref{line:greedy_k}. Then, in Line~\ref{line:greedy_z}, it assigns $\check{y}_{k^{*}}$, $\check{x}^{*}_{k^{*},i}$, and $\check{z}^{*}_{k^{*},l_{i}}$ as $y^{*}_{k^{*}}$, $x^{*}_{k^{*},i}$, and $z^{*}_{k^{*},l_{i}}$, respectively. Finally, user $k^{*}$ is moved from $\mathcal{K}'$ to $\mathcal{K}^{*}$, which is the set of served users. The procedure is repeated until the total bandwidth is used up.

We next state the following theorems for Algorithm~\ref{algorithm_greedy}.
\begin{theorem}\label{theorem_3}
    The proposed Algorithm \ref{algorithm_greedy} obtains the optimal solution to $\mathcal{P}1$.
\end{theorem}
\begin{proof}
    The proof is shown in Appendix \myref{proof_theorem_3}.
\end{proof}
\begin{theorem}\label{theorem_4}
     The proposed Algorithm \ref{algorithm_greedy} has a polynomial time complexity $O\left(K^{2}+KIL_{\max}\right)$, where $L_{\max}=\mathop{\max}\limits_{i \in \mathcal{I}}\left\{L_{i}\right\}$.
\end{theorem}
\begin{proof}
The proof is shown in Appendix \myref{proof_theorem_4}.
\end{proof}
\section{Numerical Results}

This section evaluates the performance of SLIDE through device measurements and extensive simulations. We describe the experimental setup, evaluate SLIDE under varying settings, investigate the impact of user mobility and computing capability, conduct ablation studies, and compare the running time of the proposed algorithm with general MINLP algorithms.
\subsection{Experimental Setup}\label{sec_setup}
\begin{figure}[t]
\centering
\includegraphics[width=0.33\textwidth]{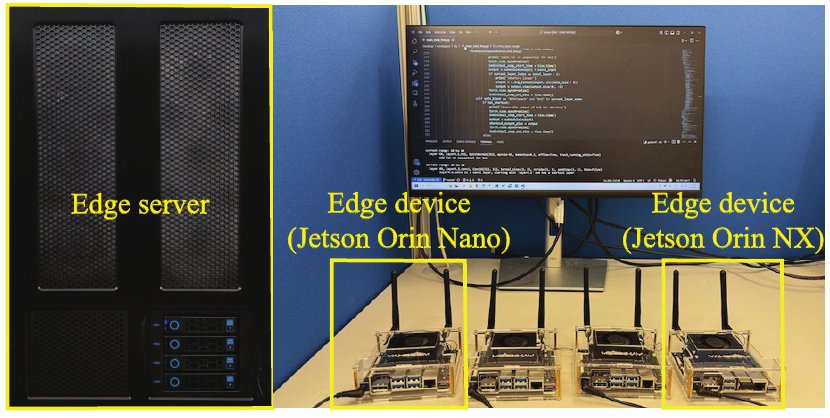}
\vspace{-0.25cm}
\caption{Experimental hardware system with an edge server (functioning as a BS) and multiple edge devices, including Jetson Orin Nano with 4 GB RAM and Jetson Orin NX with 16 GB RAM.}
\label{fig:hardware_system}
\vspace{-10pt}
\end{figure}
We adopt a hybrid testing approach for our experiments: Communication latency is assessed through simulations to consider large-scale cellular users, while computing latency is measured using real devices, as shown in Fig. \ref{fig:hardware_system}. In our simulation settings, the coverage radius and transmit power spectral density of the BS are set to 200 m and $\text{-29 dBm/Hz}$~\cite{3gpp.38.104}, respectively, with $K=$ \{60, 70, 80, 90, 100\} users uniformly distributed within the coverage area. Rayleigh fading channels are considered. The total bandwidth $B$ of the BS is \{200, 300, 400, 500, 600\} MHz, and the E2E latency requirement $\bar{T}_{k}$ ranges from 600 to 1000 ms~\cite{3gpp.22.874}. Moreover, the BS hosts $I=$ 48 AI models, which comprise CNNs and ViTs, with full-precision, 16-bit, and 8-bit variants of ResNet-18, ResNet-34, ResNet-50~\cite{he2016deep}, and DeiT-S~\cite{pmlr-v139-touvron21a}. The inference accuracy of models in $\mathcal{I}$, which varies with model structure and precision, ranges from approximately 75\% to 95\%, with lower-precision variants generally achieving lower accuracy. Each user generates one inference task, selected from 10 task types, with an inference accuracy requirement ranging from 80\% to 90\%. Each task can be served by models with one to four of the aforementioned model structures. For each user, $\mathcal{I}_{k}$ is constructed by selecting the set of models that can serve user $k$, with the desired accuracy requirements. To reflect the heterogeneity in computing capabilities, we consider two types of user devices: more powerful Jetson Orin NX devices and less powerful Jetson Orin Nano devices, with maximum GPU frequencies of 918 MHz and 624.75 MHz, respectively. The proportion of Jetson Orin Nano users is set to $\theta= \{\text{0\%, 20\%, 40\%, 60\%, 80\%, 100\%}\}$. Inference is performed with the image sample from the CIFAR-10 dataset~\cite{krizhevsky2009learning} with a batch size of 1. The inference energy budget $Q_{k}$ of user $k$ is given by $Q_{k}=\beta_{k}\bar{P}_{k}\bar{T}_{k}$, where $\bar{P}_{k}$ denotes the power, set to 10 Watt for Jetson Orin NX~\cite{nvidianx} and 5 Watt for Jetson Orin Nano~\cite{nvidianano}, and $\beta_{k}= \{\text{22\%, 24\%, 26\%, 28\%, 30\%}\}$ is the power scaling factor of user $k$.


We compare our SLIDE with the following baselines:
\begin{itemize}
    \item \textbf{Conventional DAI}: the conventional downloading and inference approach, where users start inference only after the entire model is downloaded~\cite{qu2024trimcachingICDCS}. To ensure a fair comparison, this baseline adopts the same user selection, bandwidth allocation, and computing resource allocation procedures as SLIDE by solving problem $\mathcal{P}1$ using Algorithm~\ref{algorithm_greedy}, with the only difference lying in the E2E latency model. Specifically, according to \eqref{eq_final_e2e}, the E2E latency of user $k$ in conventional DAI is given by $t_{k,L_{i}}\left(y_{k},{\bf{Z}}_{k,i}\right)= \sum\limits_{l_{i}=1}^{L_{i}}\frac{S_{l_{i}}}{y_{k}BR_{k}}+\sum\limits_{l_{i}=1}^{L_{i}}T_{k,i}\left(z_{k,l_{i}}\right)$, which is the sum of model downloading and inference latencies, corresponding to the sequential model downloading and inference.
    \item \textbf{B\&B (Branch-and-Bound)}: a general algorithm for solving MINLP problems. 
    Since both SLIDE and B\&B obtain the optimal solution to $\mathcal{P}1$, we only compare them in terms of algorithm running time.
\end{itemize}

To compare the algorithm performance, we use the served user ratio as the evaluation metric, defined as the ratio of the number of users successfully served within the latency requirements to the total number of users.
\subsection{Performance Evaluation}
This subsection evaluates the performance of SLIDE by varying $B$, $K$, $\bar{T}_{k}$, $\theta$, and $\beta_{k}$.

\begin{figure*}[t]
    \centering
	\subfigure[Served user ratio vs. $B$.]{\includegraphics[height=2.9cm, keepaspectratio]{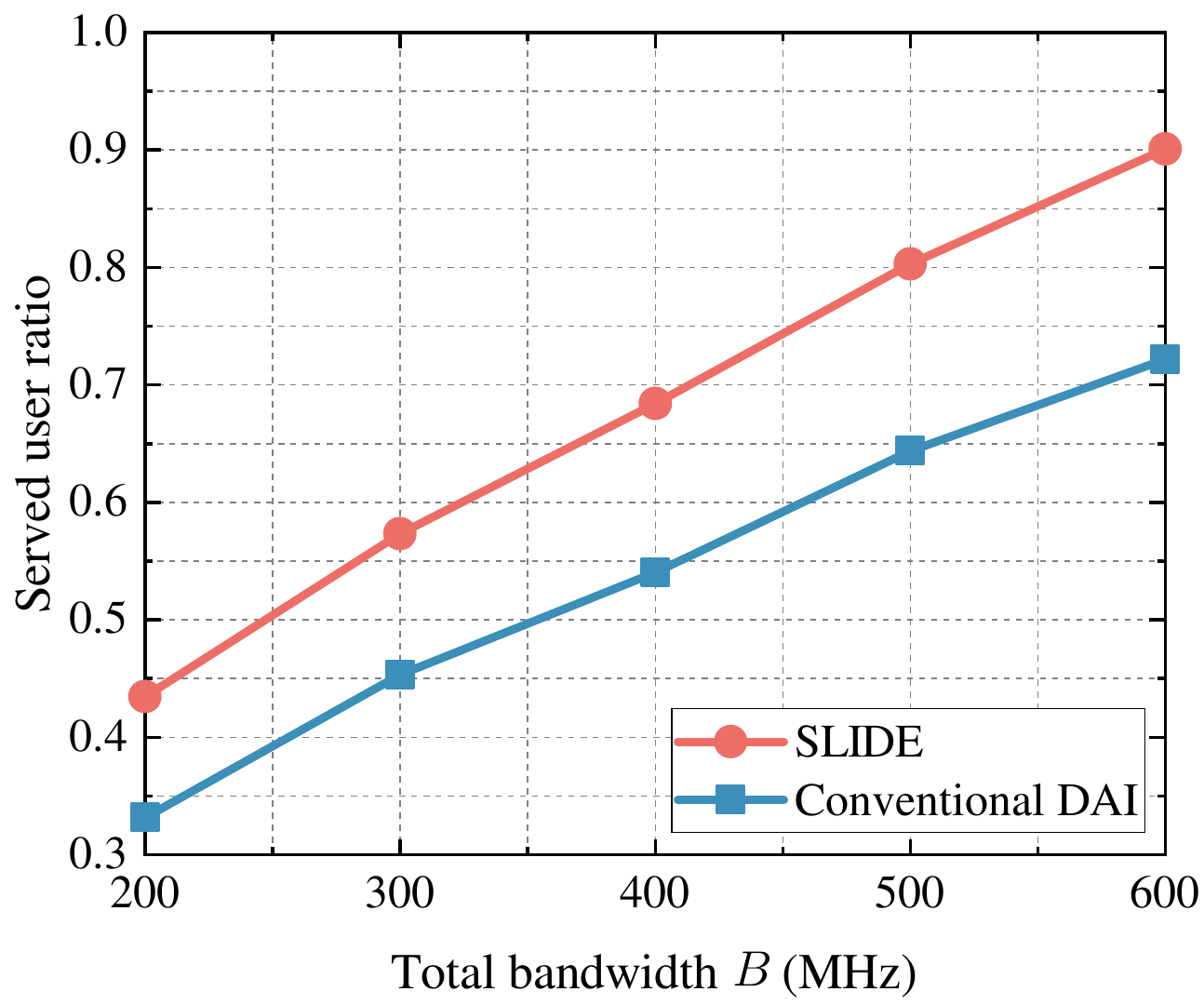}\label{fig:bw}}
	\subfigure[Served user ratio vs. $K$.]{\includegraphics[height=2.9cm, keepaspectratio]{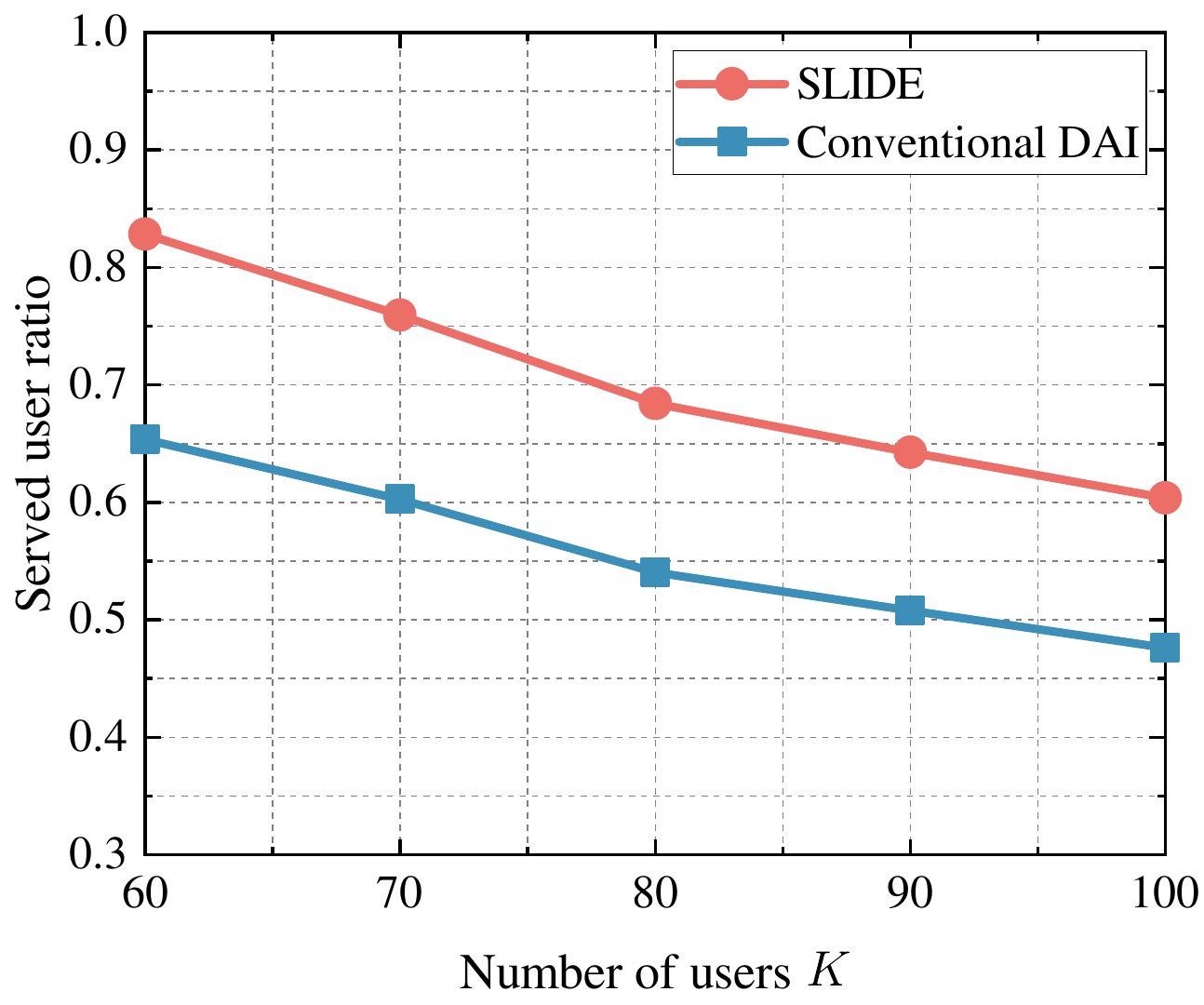}\label{fig:user}}
	\subfigure[Served user ratio vs. $\bar{T}_{k}$.]{\includegraphics[height=2.9cm, keepaspectratio]{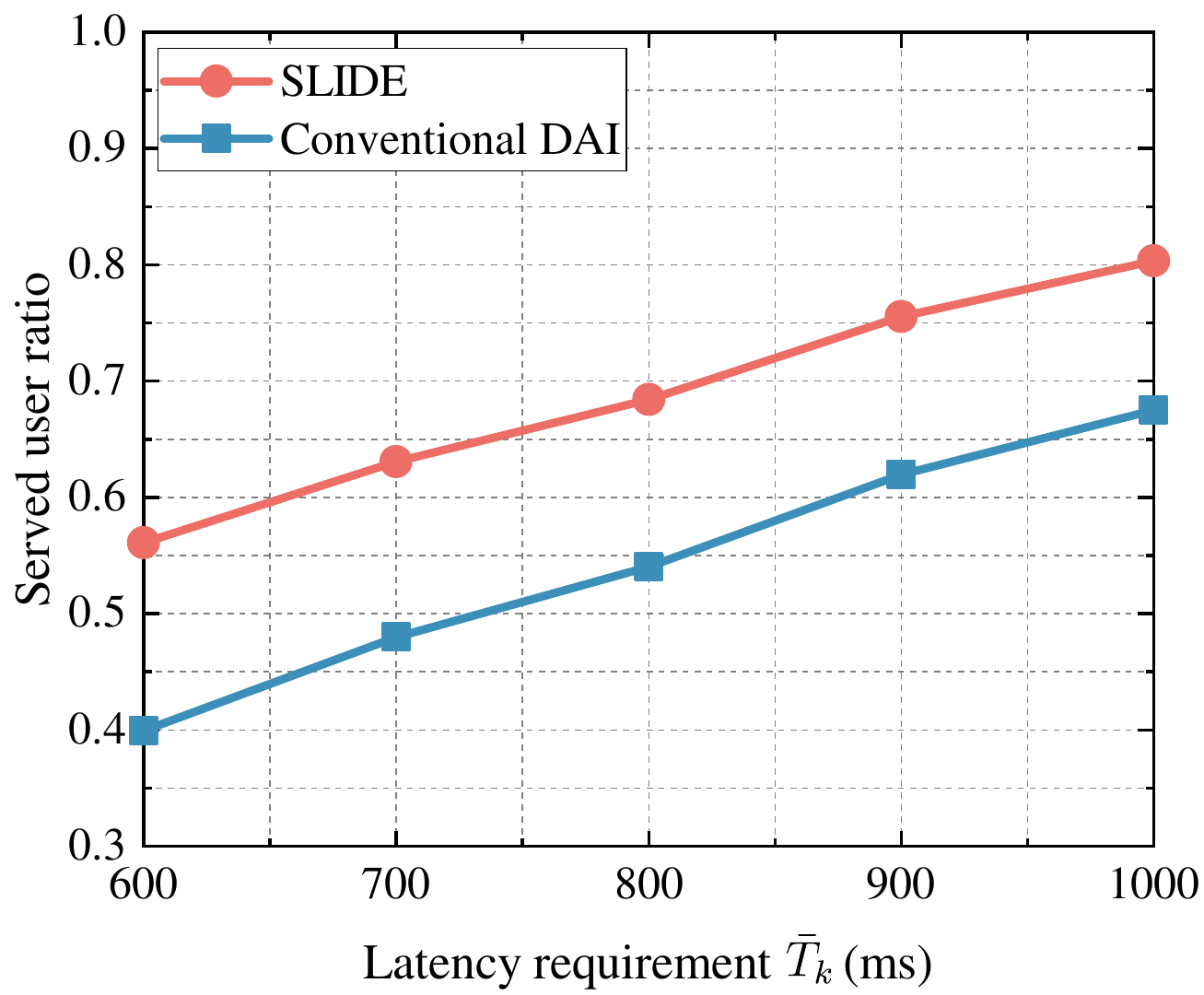}\label{fig:ddl}}
	\subfigure[Served user ratio vs. $\theta$.]{\includegraphics[height=2.9cm, keepaspectratio]{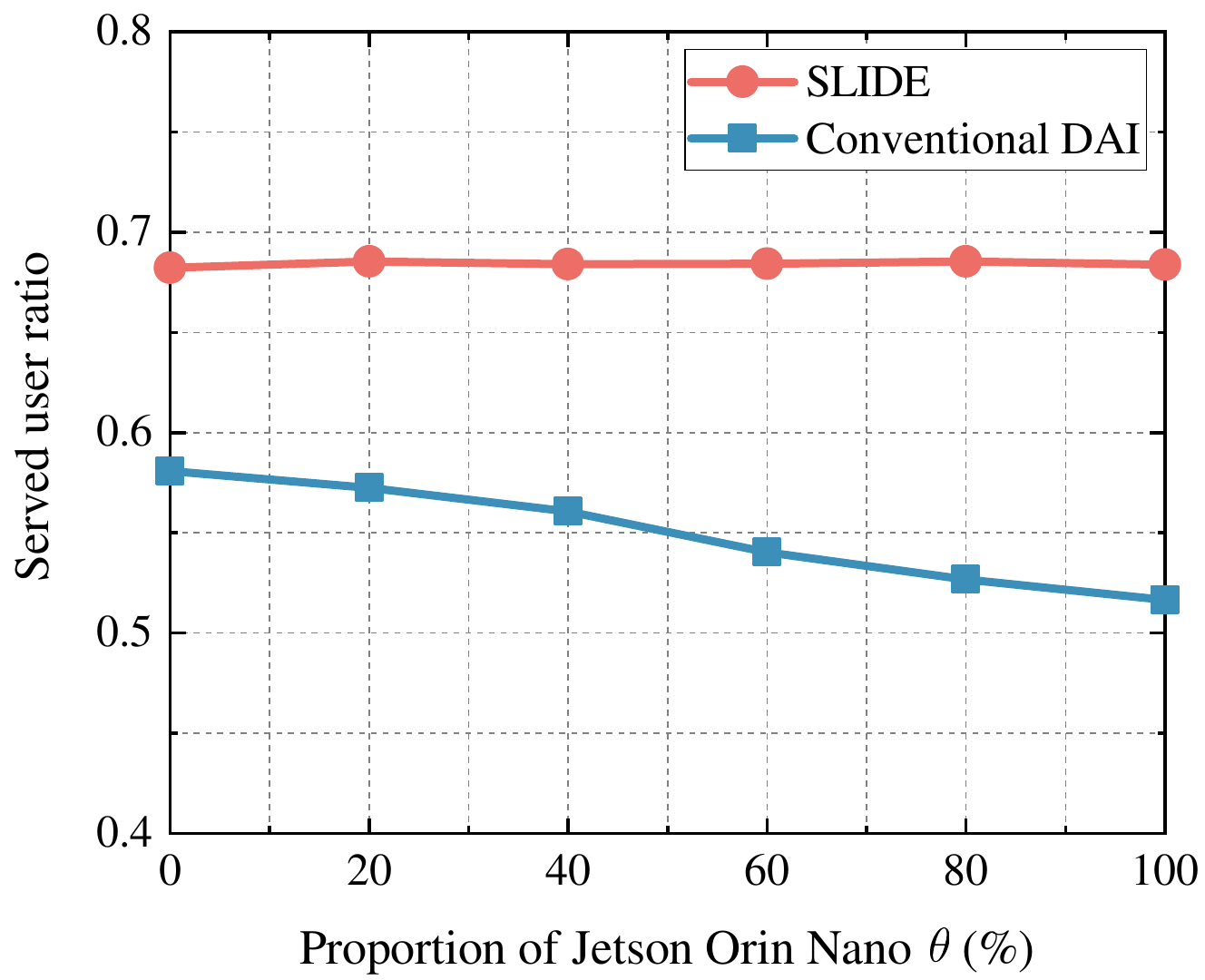}\label{fig:ratio}}
        \subfigure[Served user ratio vs. $\beta_{k}$.]{\includegraphics[height=2.9cm, keepaspectratio]{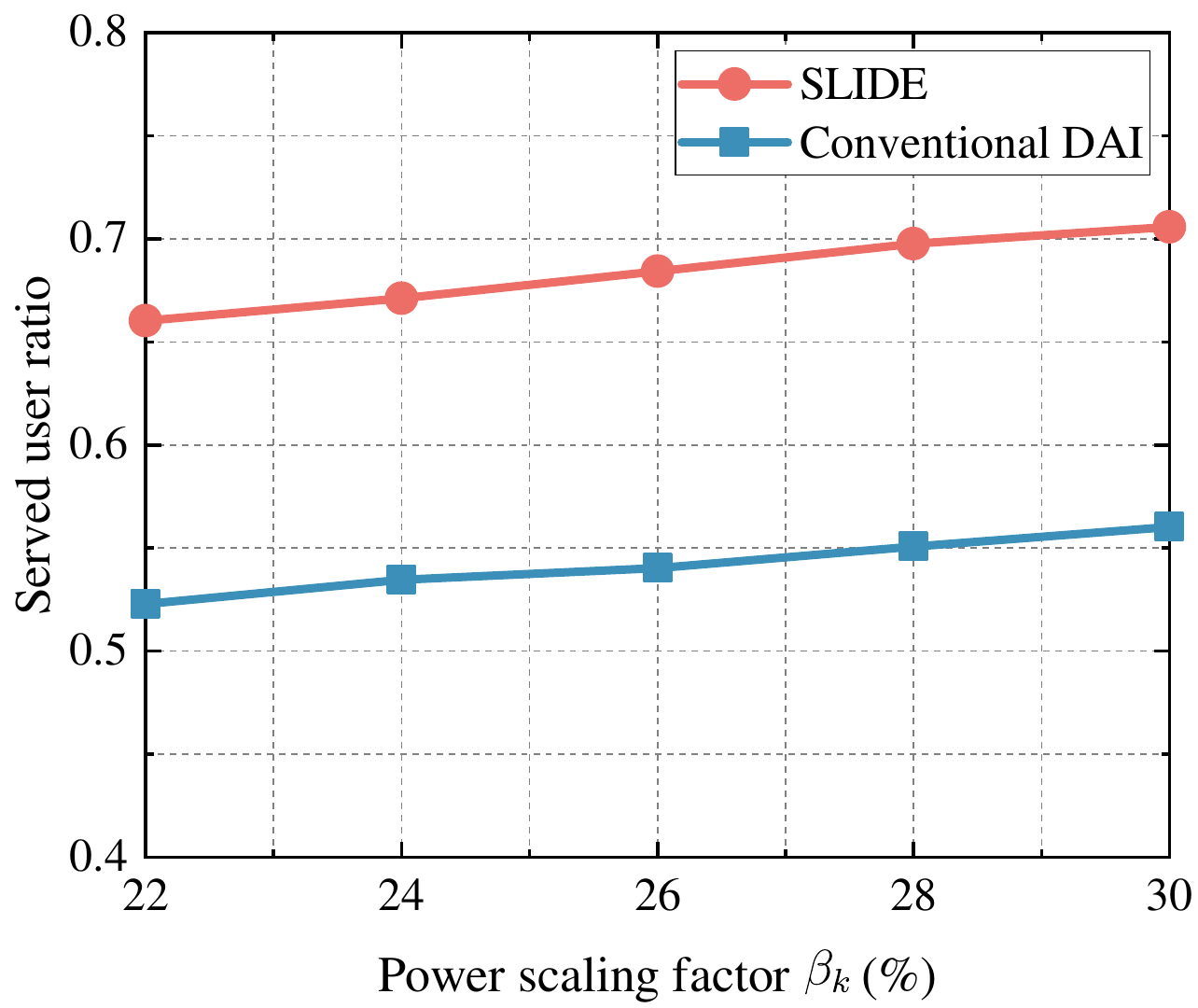}\label{fig:beta}}
        \vspace{-3pt}
 \caption{Served user ratio of SLIDE, evaluated on the Jetson Orin Nano and Jetson Orin NX running at GPU frequencies of 624.75 MHz and 918 MHz, respectively. The default values of $B$, $K$, $\bar{T}_{k}$, $\theta$, and $\beta_{k}$, are set to 400 MHz, 80, 800 ms, 60\%, and 26\%, respectively.}
 \vspace{-10pt}\label{fig:result}
\end{figure*}

Fig. \ref{fig:bw} illustrates the performance of SLIDE under varying $B$. The served user ratio of both SLIDE and conventional DAI increases almost linearly as $B$ increases. Moreover, SLIDE consistently outperforms conventional DAI, with an average 14.1\% higher served user ratio across varying $B$. Fig. \ref{fig:user} examines the performance of SLIDE as $K$ varies. The served user ratios of both SLIDE and conventional DAI decline with increasing $K$. However, SLIDE consistently outperforms conventional DAI, achieving a 14.7\% higher served user ratio on average. Notably, SLIDE can still serve over 60.4\% of users even when $K=$ 100, demonstrating its effectiveness under high user density. The performance gain of SLIDE is also evident in Fig. \ref{fig:ddl}. With varying $\bar{T}_{k}$, SLIDE achieves an average served user ratio that is 14.4\% higher than that of conventional DAI. 

Fig. \ref{fig:ratio} shows the performance of SLIDE as $\theta$ varies. On average, SLIDE improves the served user ratio by 13.5\%. As $\theta$ increases, the performance of conventional DAI degrades significantly, whereas SLIDE maintains stably. Notably, when all users are Jetson Orin Nano devices with limited computing capabilities, SLIDE can serve 68.4\% of users, which is 16.7\% higher than conventional DAI. These results demonstrate that SLIDE sustains high task throughput even with more users equipped with limited computing capabilities. Fig. \ref{fig:beta} illustrates the impact of $\beta_k$, which controls the energy budget, on the served user ratio. As anticipated, increasing $\beta_k$ (equivalent to raising $Q_{k}$) enhances the served user ratio. On average, SLIDE achieves a 14.2\% higher served user ratio than conventional DAI as $\beta_{k}$ varies.

Fig. \ref{fig:quantization} evaluates the performance of SLIDE under different model libraries with various precision levels, including full precision, 16-bit quantization, 8-bit quantization, and a mix of them (which is the default setting in our simulations). With varying $B$ and $K$, both Fig.~\ref{fig:quantization_bw} and Fig.~\ref{fig:quantization_user} demonstrate that SLIDE with mixed-precision models significantly enhances the served user ratio compared with only full-precision or 16-bit models. Moreover, the performance advantage over the 8-bit model library becomes especially notable when total bandwidth is sufficient. These results demonstrate that SLIDE can effectively utilize a model library with diverse precision versions, enabling efficient model provisioning for users.
\begin{figure}[t]
    \centering
        \quad\includegraphics[width=0.4\textwidth]{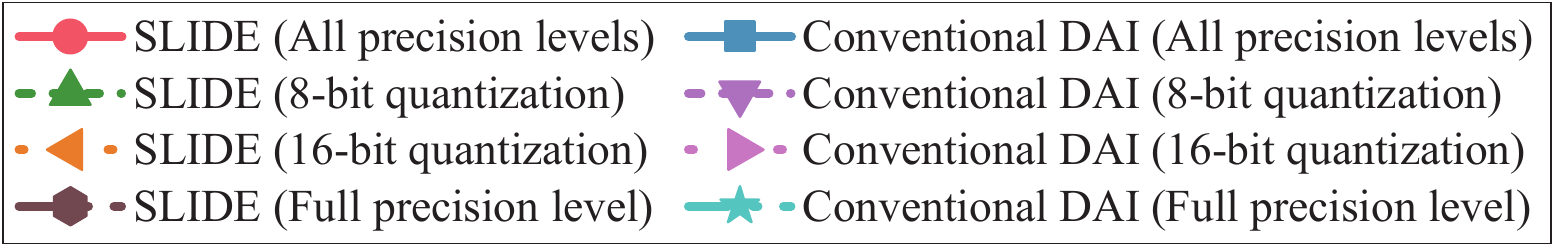}\\
	\subfigure[Served user ratio of SLIDE vs. $B$ under different model libraries.]{\includegraphics[height=3.4cm, keepaspectratio]{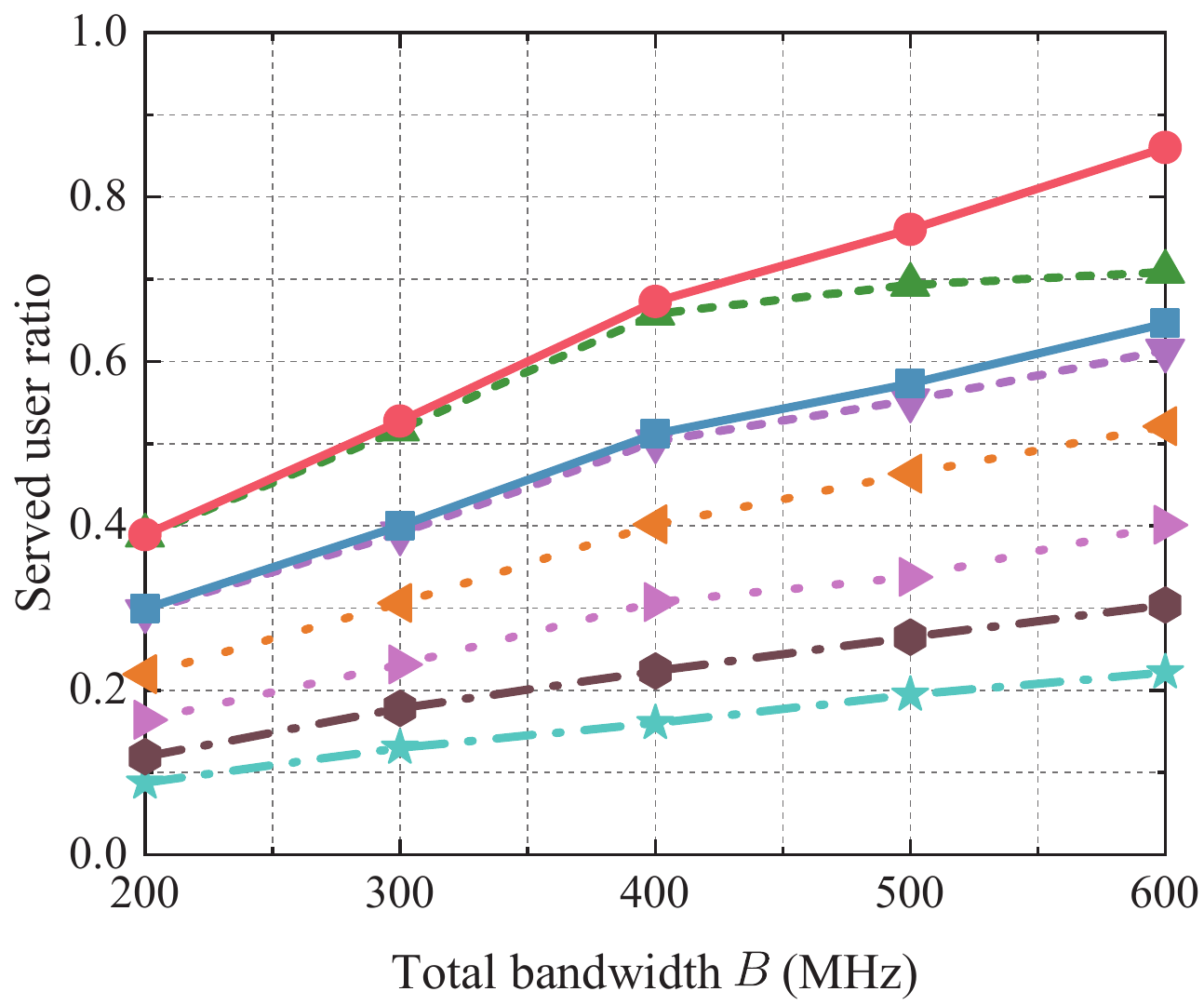}\label{fig:quantization_bw}}
	\quad
	\subfigure[Served user ratio of SLIDE vs. $K$ under different model libraries.]{\includegraphics[height=3.4cm, keepaspectratio]{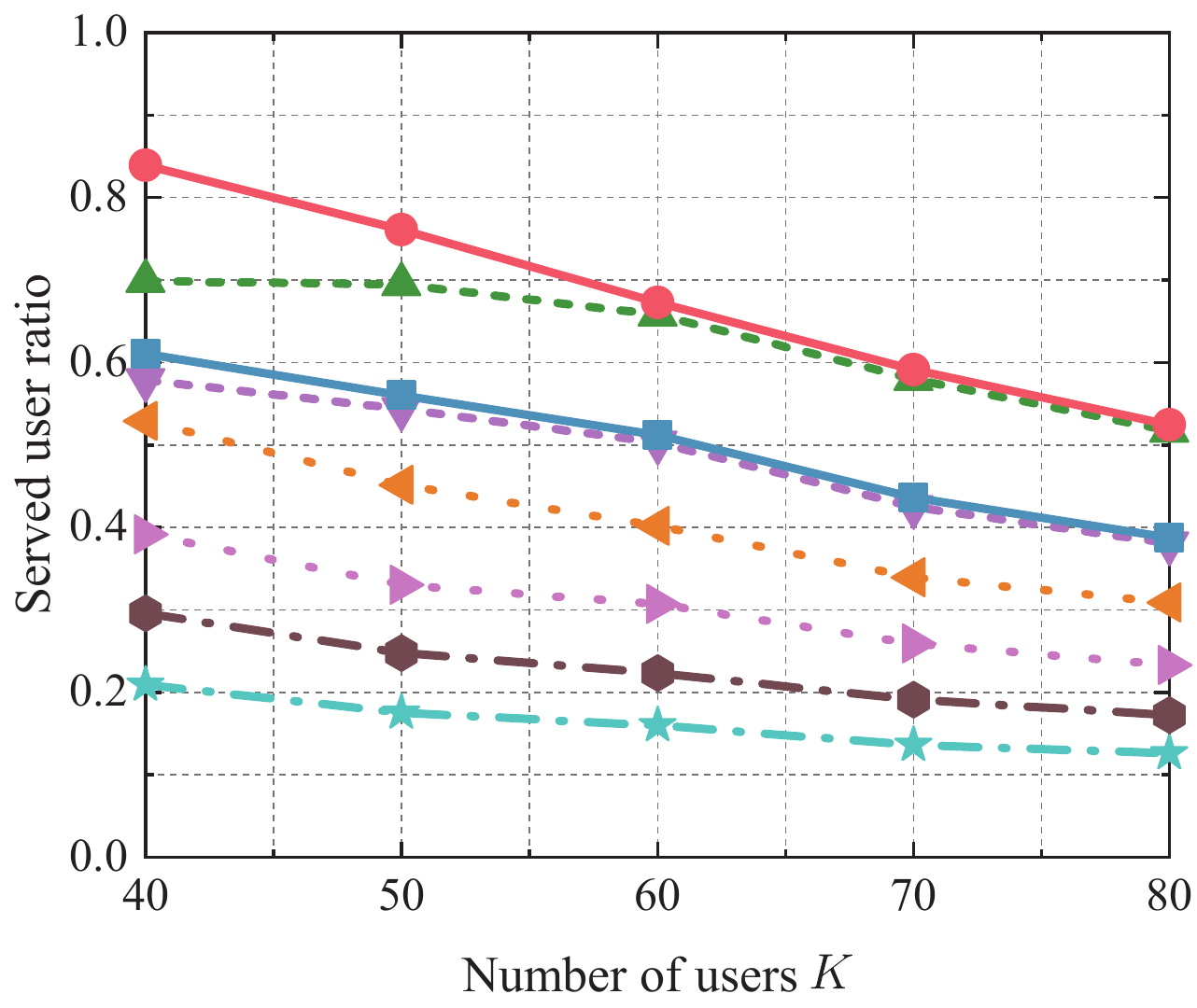}\label{fig:quantization_user}}
 \caption{Performance of SLIDE under different model libraries, where the default for $K$ is set to 60, and the inference accuracy requirement of users ranges from 87\% to 95\%. The default values of GPU frequencies, $B$, $\bar{T}_{k}$, $\theta$, and $\beta_{k}$ are consistent with Fig.~\ref{fig:result}. }\label{fig:quantization}
\end{figure}

\subsection{Impact of User Mobility}
\begin{figure}[t]
    \centering
	\subfigure[Served user ratio vs. $B$ in mobile scenarios.]{\includegraphics[height=3.4cm, keepaspectratio]{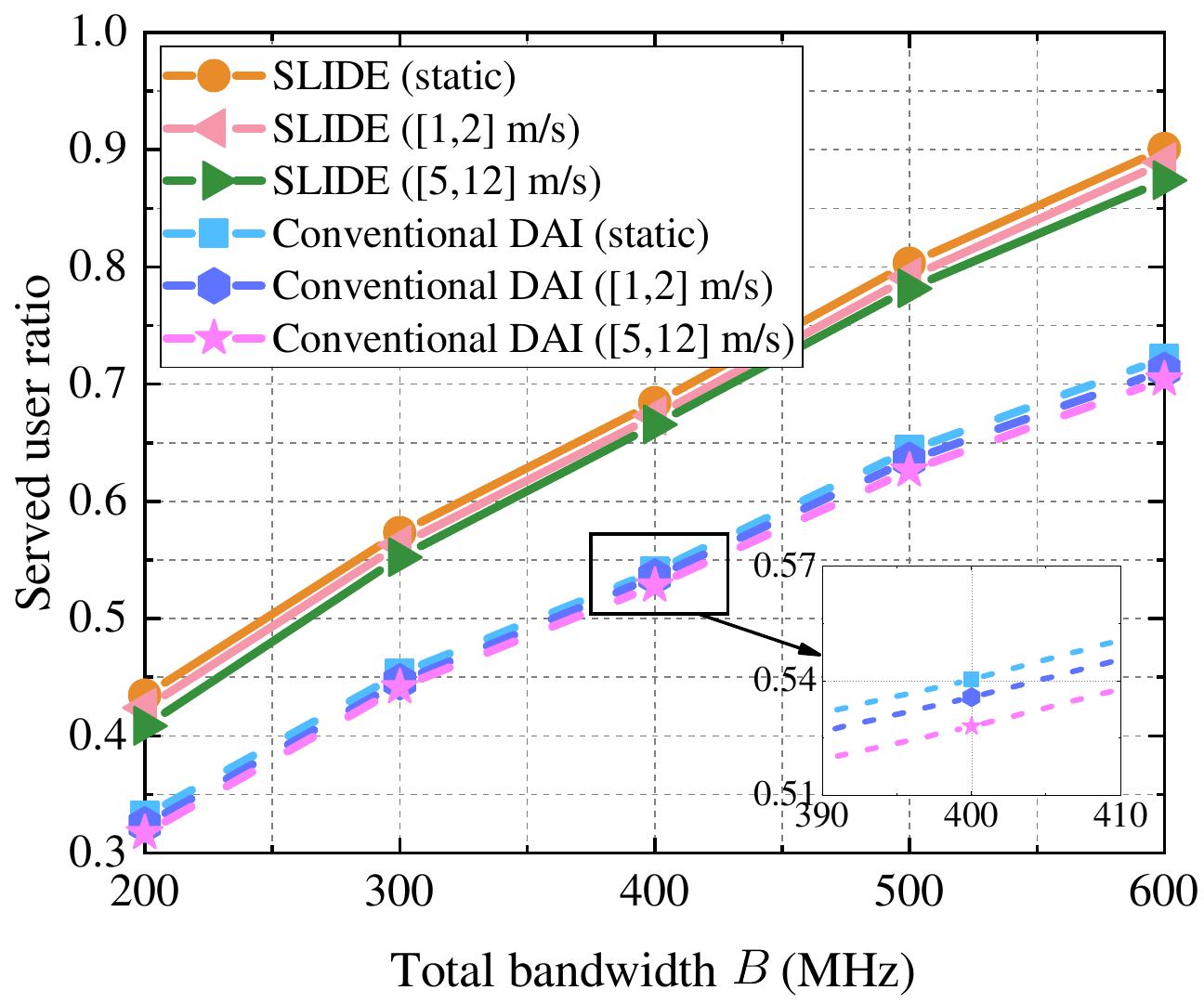}\label{fig:mobility_bw}}
	\quad
	\subfigure[Served user ratio vs. $K$ in mobile scenarios.]{\includegraphics[height=3.4cm, keepaspectratio]{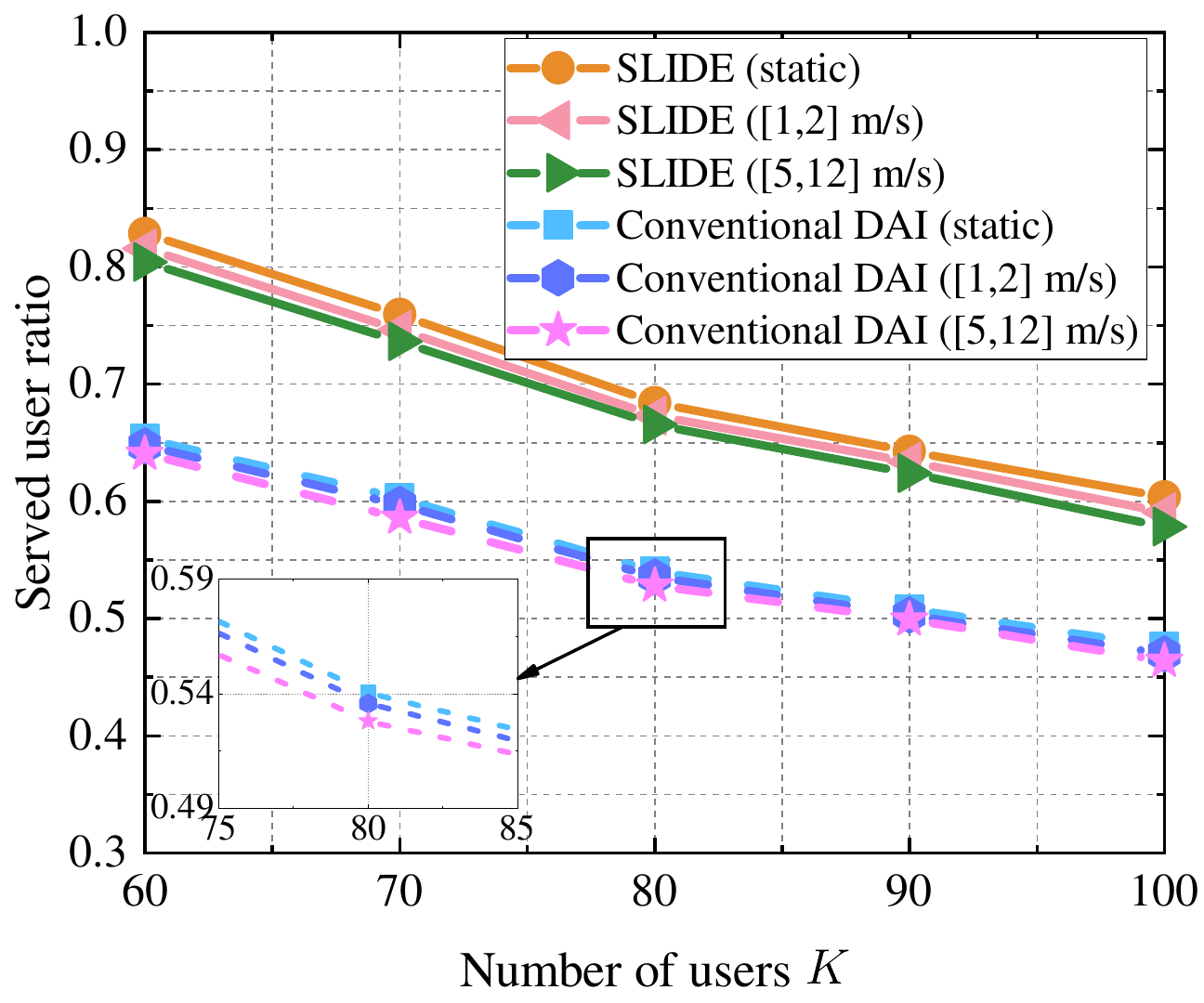}\label{fig:mobility_user}}
    \vspace{-3pt}
 \caption{Performance of SLIDE in mobile scenarios, where the default values of GPU frequencies, $B$, $K$, $\bar{T}_{k}$, $\theta$, and $\beta_{k}$ are consistent with Fig. \ref{fig:result}.}
 \vspace{-10pt}
 \label{fig:mobility}
\end{figure}
Fig. \ref{fig:mobility} analyzes the impact of user mobility on the served user ratio in both static and dynamic scenarios with slow (1–2~m/s) and fast (5–12 m/s) user mobility. As shown in Fig.~\ref{fig:mobility_bw}, the served user ratio of SLIDE decreases by 1.1\% and 2.3\% under slow and fast mobility, respectively; in Fig.~\ref{fig:mobility_user}, the reductions are 1.2\% and 2.2\%, respectively. These results demonstrate that our algorithm remains robust under user mobility, although the resource allocation is determined using the initial user locations before they move.


\subsection{Impact of Device Computing Capability}
Fig. \ref{fig:freq} investigates the impact of user device computing capability on the performance of SLIDE. We evaluate both conventional DAI and SLIDE with Jetson Orin NX under two GPU frequency settings: 306 MHz and 918 MHz. When the GPU frequency is reduced from 918 MHz to 306 MHz, the served user ratio of conventional DAI declines by 4.8\%, 4.5\%, and 4.5\% in Figs. \ref{fig:freq_bw}, \ref{fig:freq_user}, and \ref{fig:freq_beta}, respectively, while the corresponding degradations of SLIDE are only 1.3\%, 1.3\%, and 1.4\%. This implies that SLIDE, by overlapping communications and computing, can effectively utilize computing resources of devices and maintain robust performance across devices with heterogeneous computing capabilities. 


\begin{figure*}[t]
    \centering
	\subfigure[Served user ratio vs. $B$ at different GPU frequencies.]{\includegraphics[height=3.4cm, keepaspectratio]{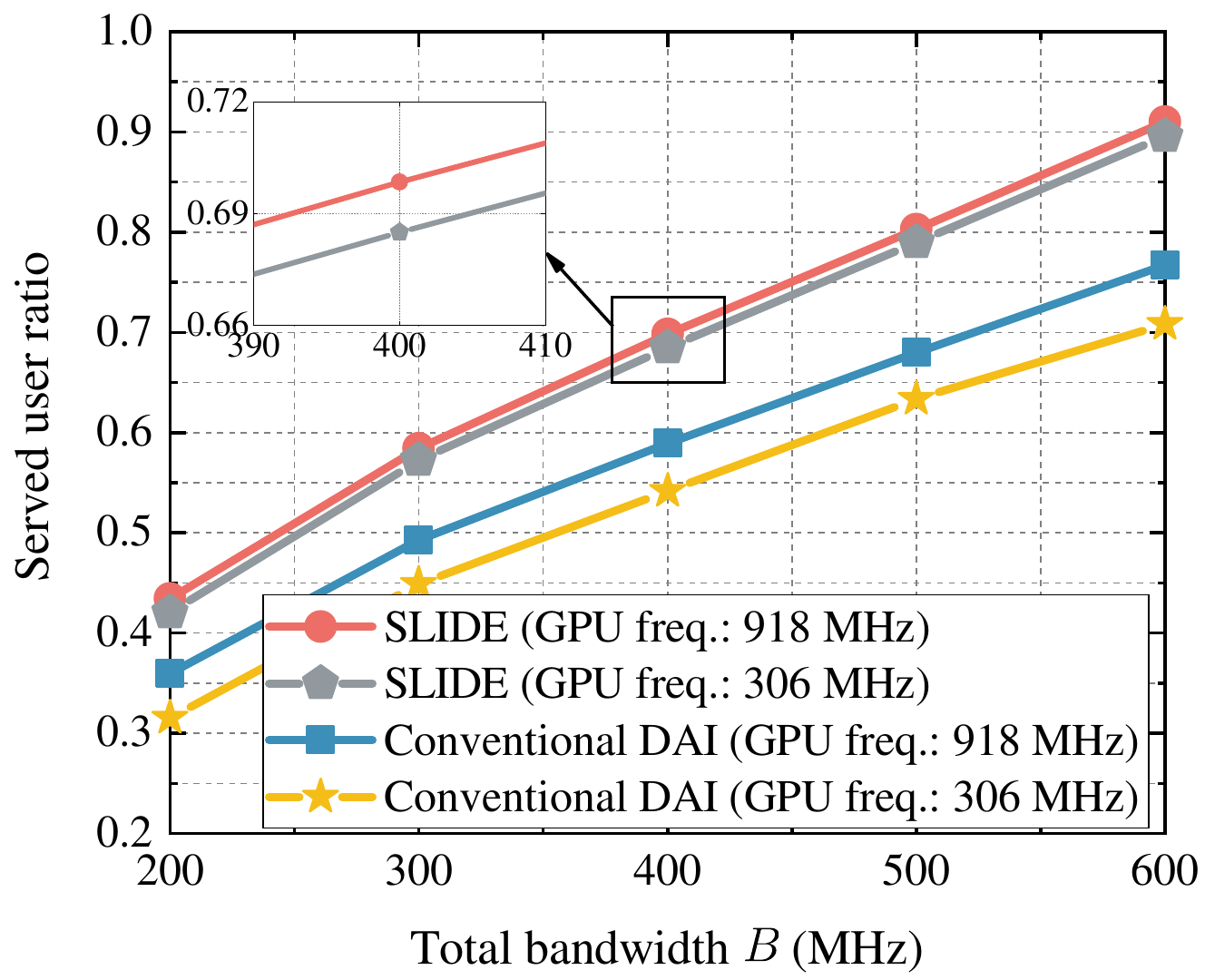}\label{fig:freq_bw}}
        \quad
	\subfigure[Served user ratio vs. $K$ at different GPU frequencies.]{\includegraphics[height=3.4cm, keepaspectratio]{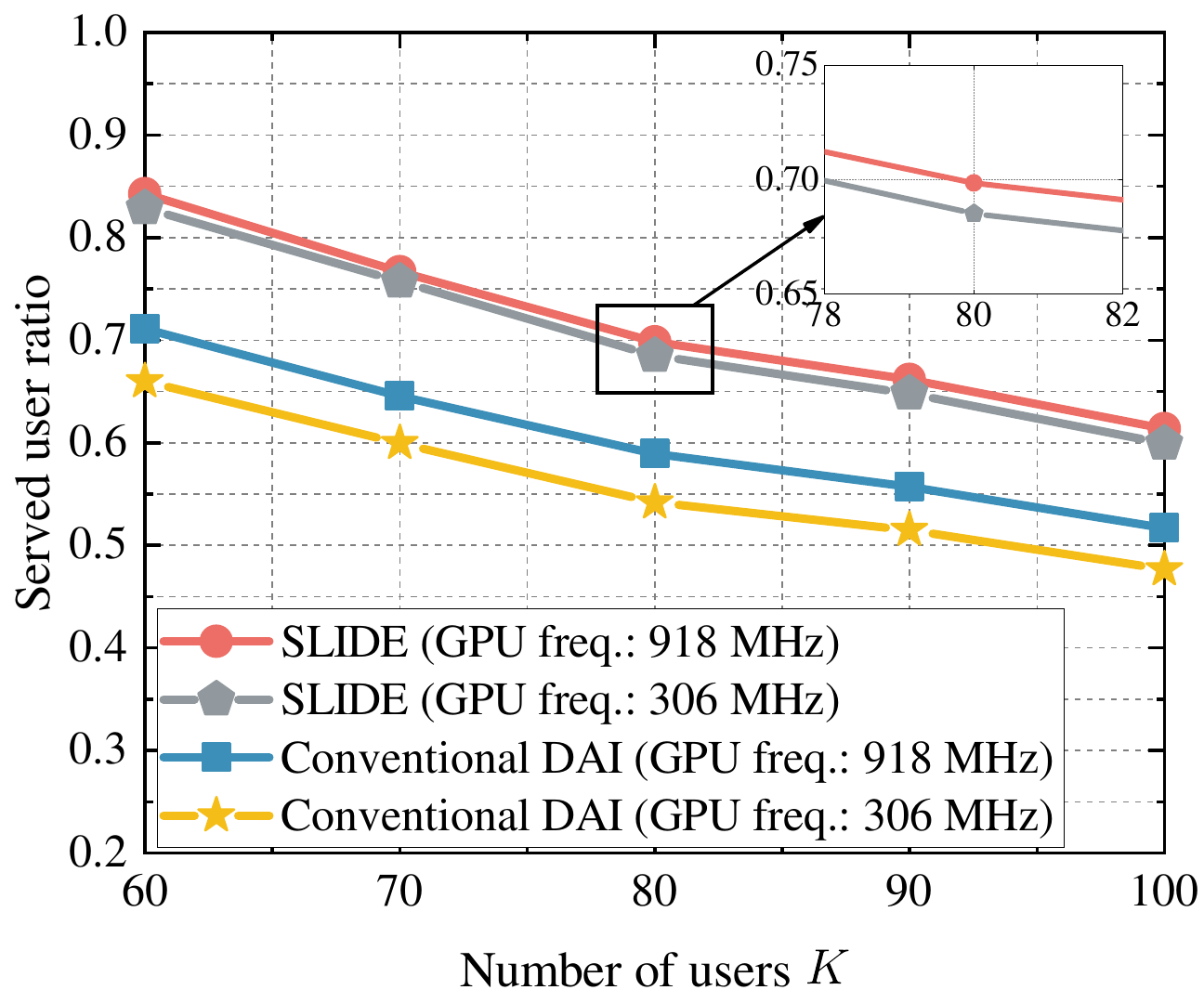}
    \label{fig:freq_user}}
	\quad
	\subfigure[Served user ratio vs. $\beta_{k}$ at different GPU frequencies.]{\includegraphics[height=3.4cm, keepaspectratio]{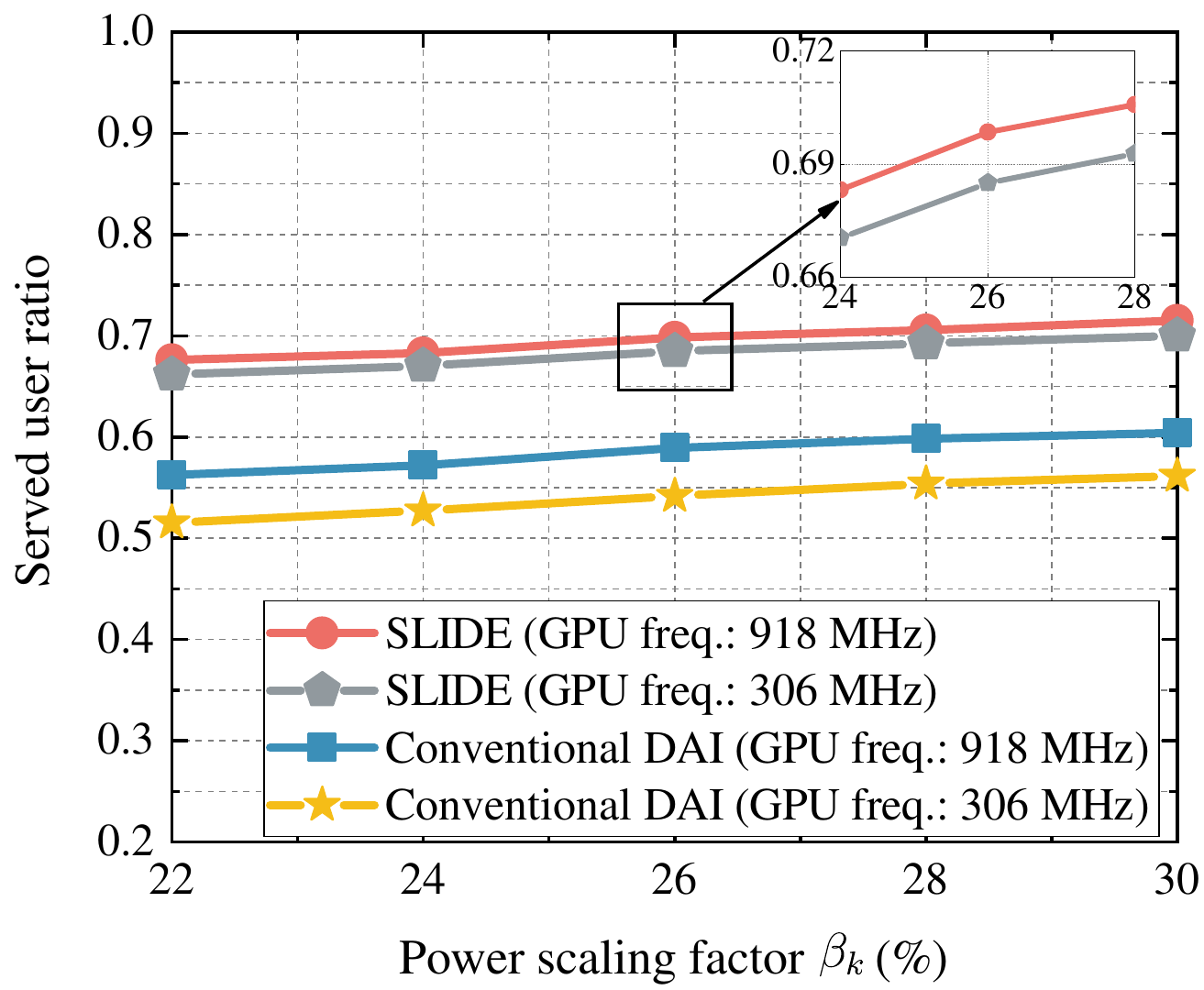}
    \label{fig:freq_beta}}
    \vspace{-3pt}
 \caption{Performance comparison of SLIDE and conventional DAI on Jetson Orin NX ($\theta = 0$) under different computing capabilities, with the GPU frequency set to 306 MHz and 918 MHz, respectively. The default values of $B$, $K$, $\bar{T}_{k}$, and $\beta_{k}$ follow those in Fig. \ref{fig:result}.}\label{fig:freq}
 \vspace{-10pt}
\end{figure*}

\subsection{Ablation Study}
This subsection analyzes the impact of model provisioning, bandwidth allocation, and computing resource allocation in SLIDE through ablation experiments. Fig. \ref{fig:ablation} compares SLIDE with the counterparts under equal bandwidth allocation, greedy-based model provisioning, and equal-energy computing resource allocation. In equal bandwidth allocation, the total bandwidth $B$ is evenly distributed to $K$ sub-channels, with each sub-channel assigned to a single user with bandwidth $\frac{B}{K}$. In greedy-based model provisioning, the BS provisions user $k$ with the model with the minimum model size in $\mathcal{I}_{k}$. In equal-energy computing resource allocation, each served user fully utilizes their energy budget, and an equal amount of energy is allocated to each layer.
\begin{figure*}[t]
    \centering
	\subfigure[Served user ratio vs. $B$ in the ablation study.]{\includegraphics[height=3.4cm, keepaspectratio]{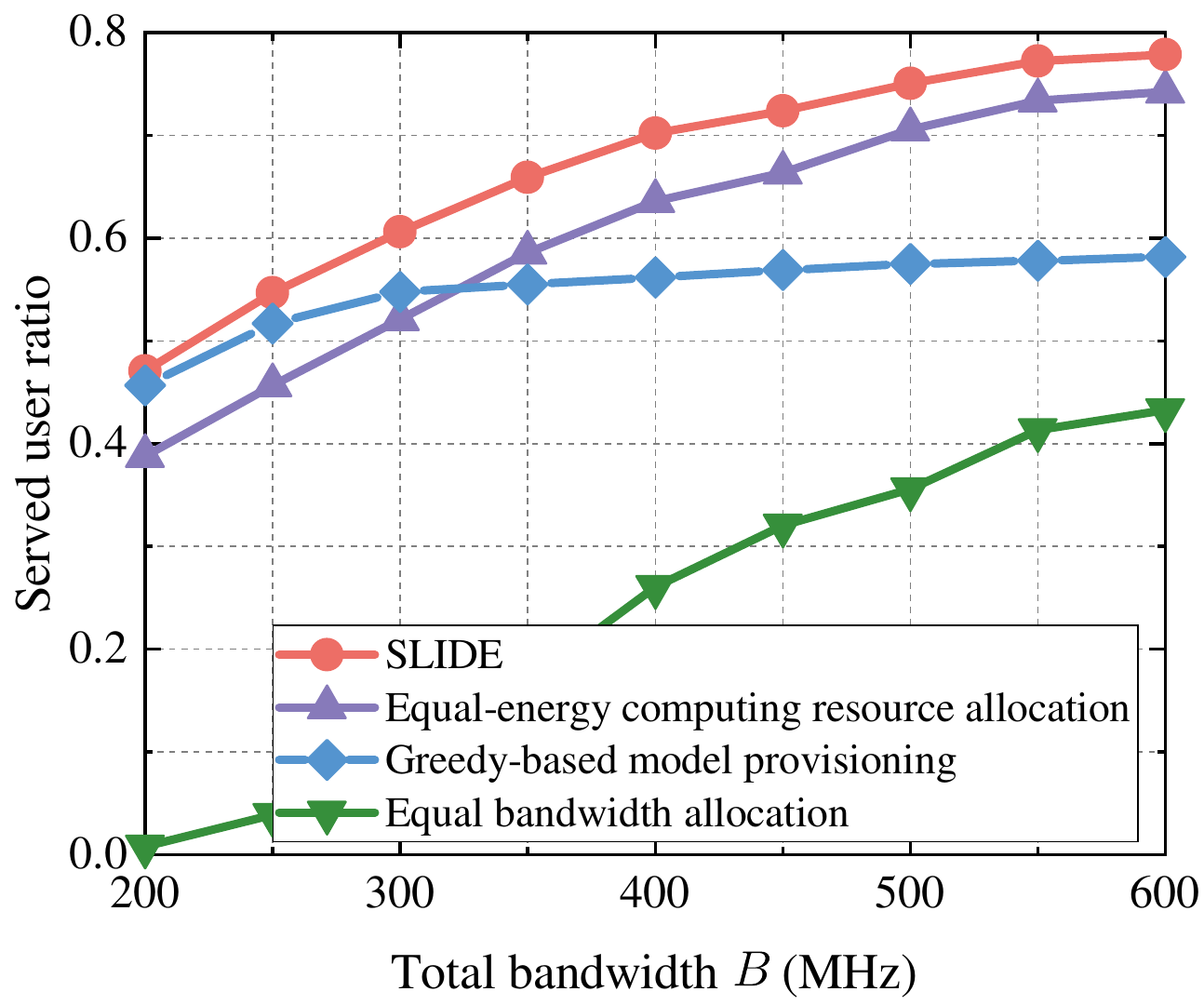}\label{fig:ablation_bw}}
	\quad
    \subfigure[Served user ratio vs. $K$ in the ablation study.]{\includegraphics[height=3.4cm, keepaspectratio]{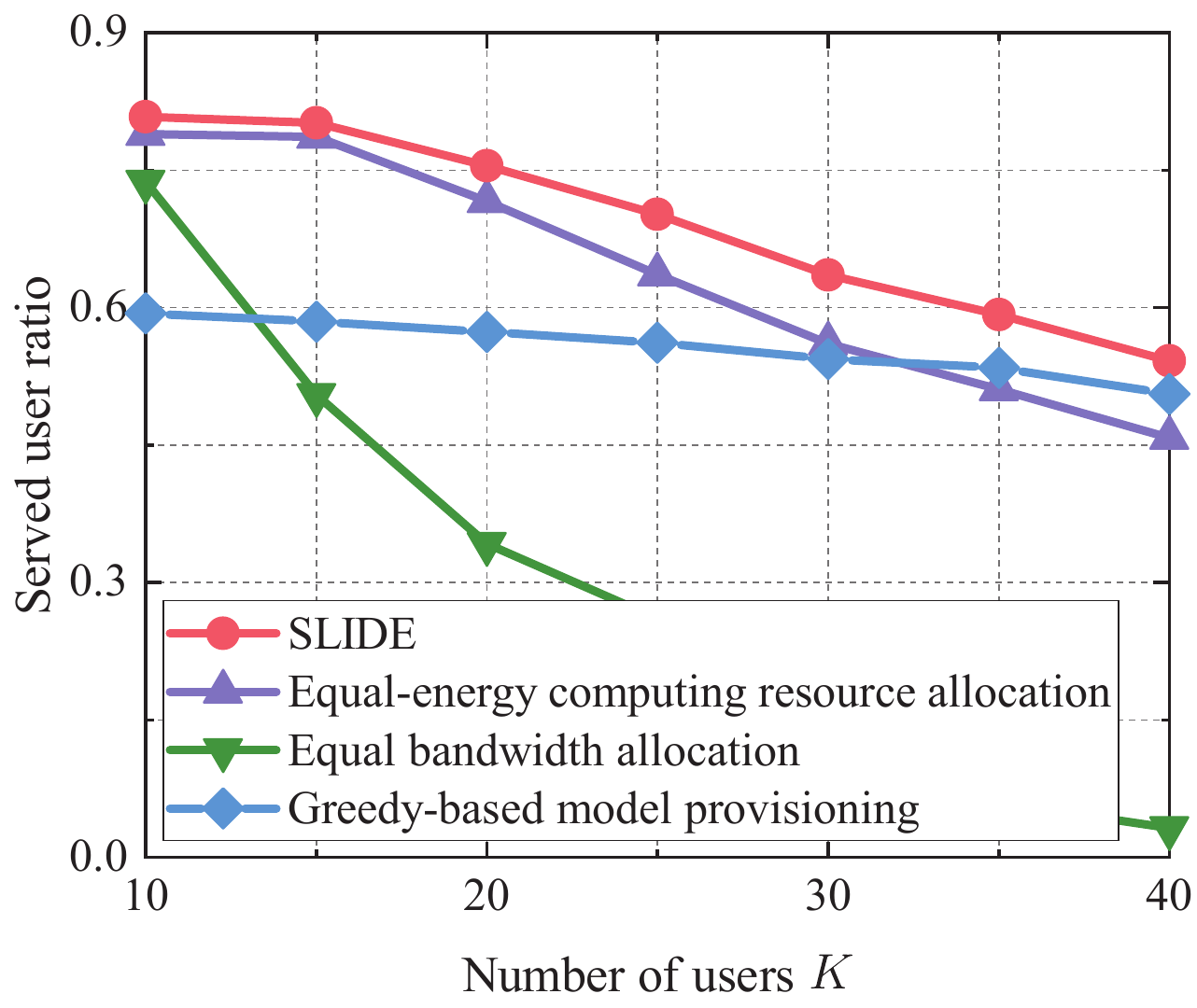}\label{fig:ablation_user}}
	\quad
	\subfigure[Served user ratio vs. $\beta_{k}$ in the ablation study.]{\includegraphics[height=3.4cm, keepaspectratio]{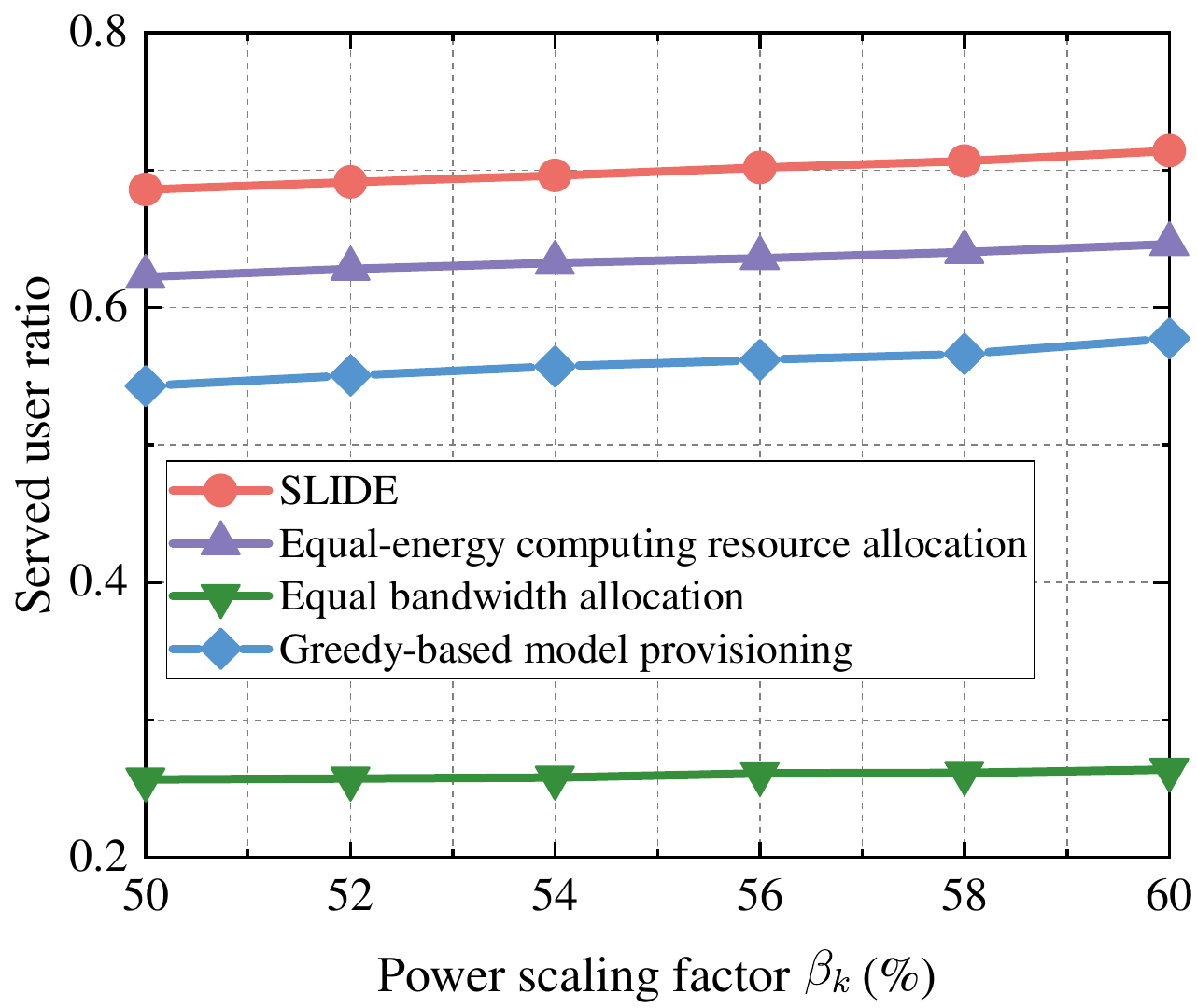}\label{fig:ablation_beta}}
    \vspace{-3pt}
 \caption{Ablation study on spectrum bandwidth allocation, model provisioning, and computing resource allocation. The default values of GPU frequencies, $B$, and $\theta$ follow those in Fig.~\ref{fig:result}, and the defaults for $K$, $\bar{T}_{k}$, and $\beta_{k}$ are set to 25, 300 ms, and 56\%, respectively.}
 \vspace{-10pt}
 \label{fig:ablation}
\end{figure*}


Fig. \ref{fig:ablation_bw} demonstrates that SLIDE consistently surpasses the equal-bandwidth-allocation (EBA), greedy-based-model-provisioning (GBMP), and equal-energy-computing-resource-allocation (EECRA) algorithms across varying $B$. Specifically, SLIDE improves the served user ratio by 43.6\%, 11.9\%, and 6.4\% on average compared with the EBA, GBMP, and EECRA methods, respectively. Moreover, at $B =$ 200 MHz, SLIDE provides a 46.3\% higher served user ratio than the EBA algorithm, which nearly fails to serve any users. This indicates that SLIDE can significantly enhance the served user ratio compared with the EBA method when the communication resource is limited. At $B =$ 600 MHz, SLIDE provides a 19.7\% improvement over GBMP. This is because SLIDE enhances the served user ratio by balancing the communication cost and the computing workload of models to perform more efficient model provisioning, whereas GBMP only minimizes communication cost without considering computing workload.

A similar performance benefit can also be observed in Fig.~\ref{fig:ablation_user}. On average, SLIDE improves the served user ratio by 39.7\%, 13.4\%, and 5.4\% over the EBA, GBMP, and EECRA methods, respectively. Fig. \ref{fig:ablation_beta} also highlights the superior performance advantage of SLIDE as $\beta_{k}$ increases. Compared with the EBA, GBMP, and EECRA methods, SLIDE achieves an average served user ratio improvement of 44.0\%, 14.0\%, and 6.5\%, respectively. At $\beta_k$ = 60\%, SLIDE outperforms the EECRA method by 6.8\%. These results underscore the importance of dedicated computing resource allocation during inference.

\subsection{Algorithm Running Time Comparisons}
\begin{figure}[t]
    \centering
	\subfigure[Algorithm running time vs. $K$.]{\includegraphics[height=3.4cm, keepaspectratio]{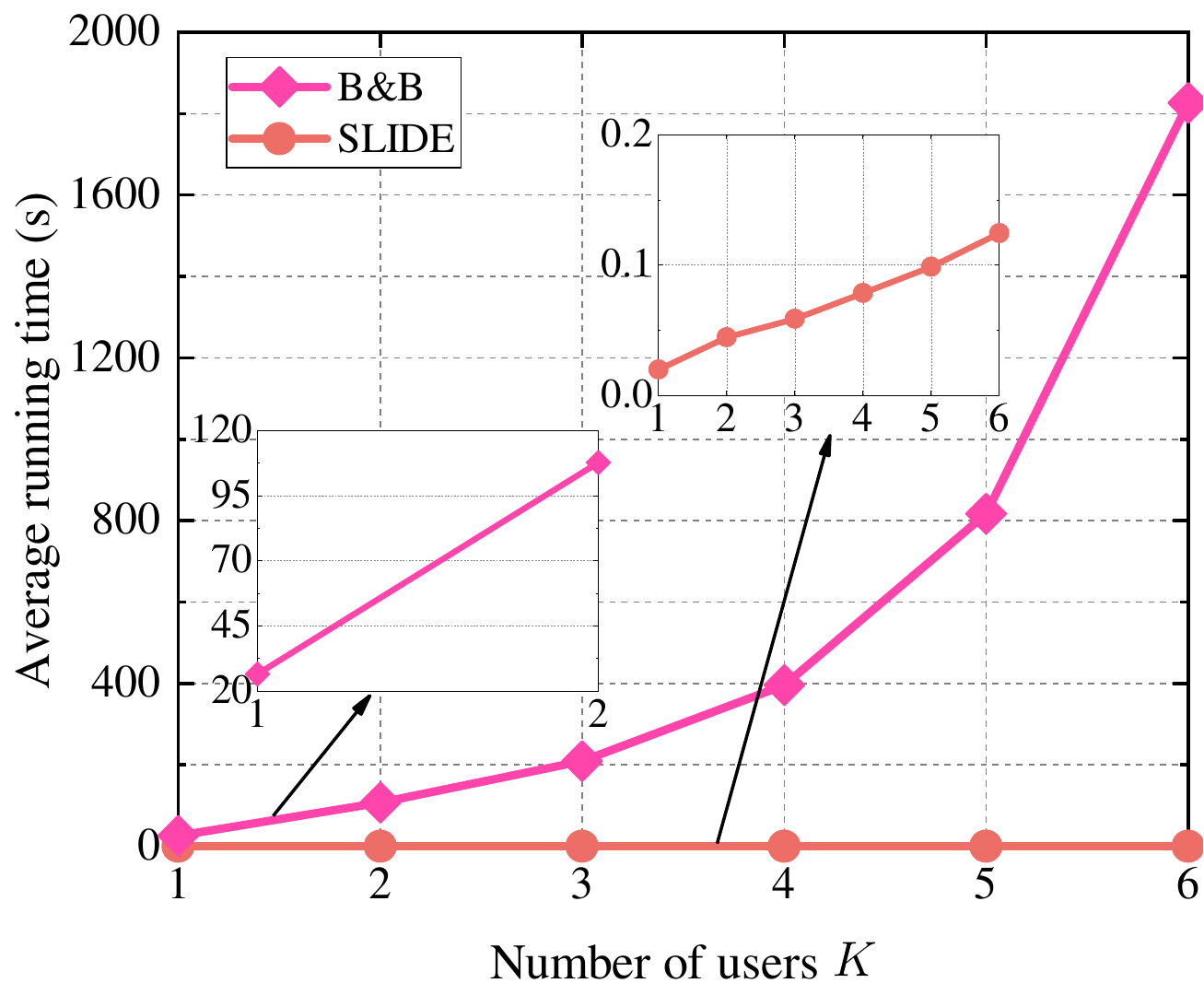}\label{fig:running_user}}
        \quad
	\subfigure[Algorithm running time vs. $I$.]{\includegraphics[height=3.4cm, keepaspectratio]{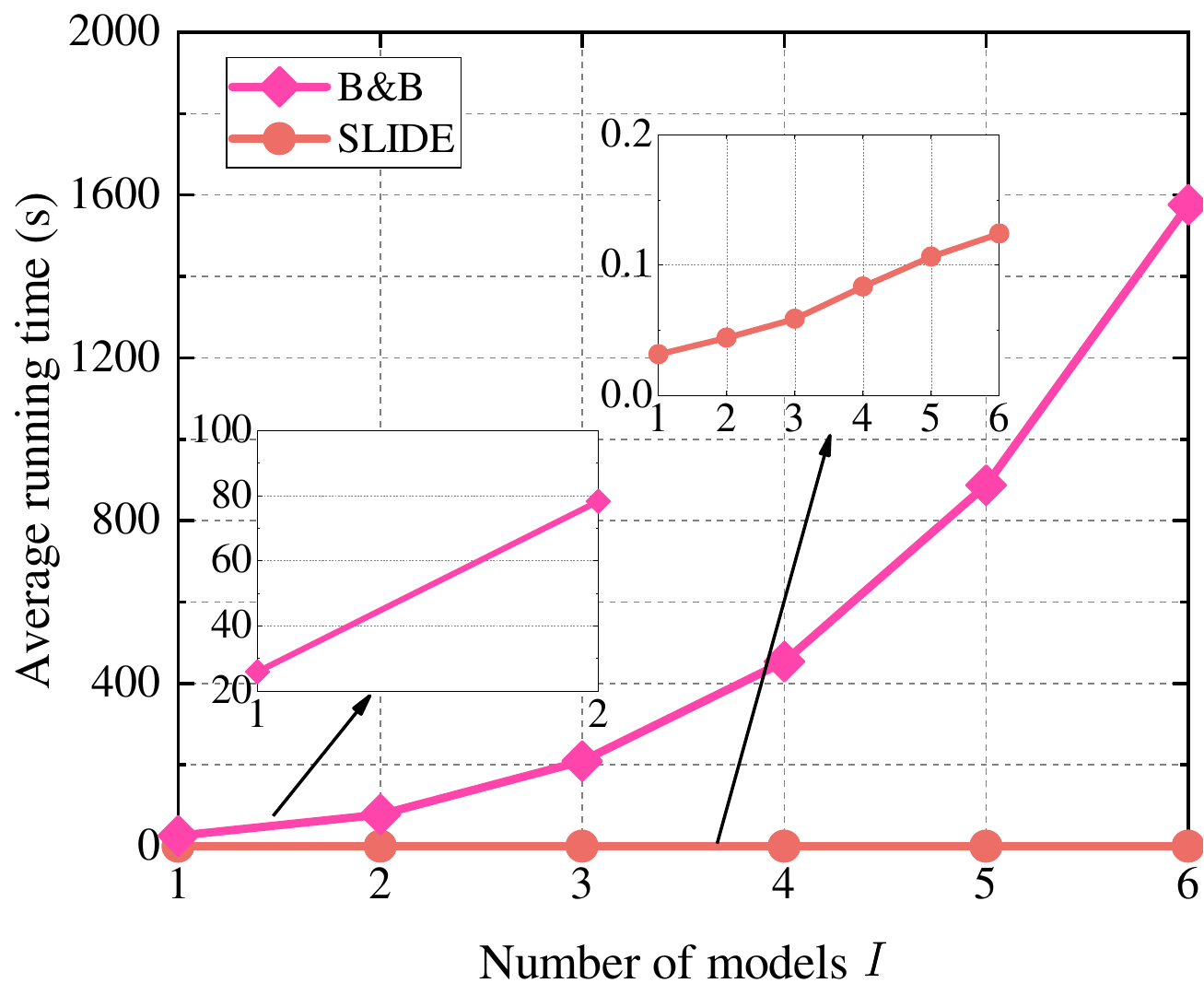}\label{fig:running_model}}
    \vspace{-7pt}
 \caption{Running time comparison between the proposed algorithm and the B\&B algorithm using Jetson Orin Nano under varying $K$, with $B =$ 200 MHz, $\bar{T}_{k}=$ 800 ms, and $\theta=1$. Both the default values of $K$ and $I$ are set to 3, and the defaults for the GPU frequency of the Jetson Orin Nano and $\beta_{k}$ are the same as those in Fig. \ref{fig:result}.}
 \vspace{-10pt}
 \label{fig:running}
\end{figure}

Fig.~\ref{fig:running} compares the running time between the proposed algorithm and the B\&B algorithm. To run the B\&B algorithm efficiently, both $K$ and $I$ are set to \{1, 2, 3, 4, 5, 6\}. As shown in Fig.~\ref{fig:running}, the proposed algorithm significantly outperforms the B\&B algorithm in terms of running time. The running time of the proposed algorithm increases nearly linearly as $K$ and $I$ grow, while the B\&B algorithm exhibits exponential growth. On average, the proposed algorithm is 5,893 times and 5,442 times faster than the B\&B algorithm in Fig.~\ref{fig:running_user} and Fig.~\ref{fig:running_model}, respectively. Specifically, when $K=$ 6 and $I=$ 6, the proposed algorithm is 14,689 times and 12,697 times faster than the B\&B algorithm. These results indicate the efficiency of the proposed algorithm in handling large-scale optimization problems.
\section{Conclusions}
In this paper, we proposed a simultaneous model downloading and inference (SLIDE) framework to maximize inference task throughput for users in wireless networks. This framework allows end users to perform inference with the downloaded model layers while concurrently receiving the remaining layers. To optimize model provisioning, spectrum bandwidth allocation, and computing resource allocation for multi-user SLIDE systems, we formulated a task throughput maximization problem under end-to-end latency, communication resource, and energy constraints. To efficiently solve this problem, we transformed the original problem into solving the spectrum bandwidth allocation problem, followed by determining the optimal model provisioning and computing resource allocation based on its derived solution. We then designed an efficient algorithm that obtains the optimal solution in polynomial time. The proposed algorithm demonstrates a significant performance gain compared with the conventional DAI framework without overlapping model downloading and inference in wireless networks.

\nocite{nesterov2013introductory}

\bibliographystyle{IEEEtran}
\bibliography{IEEEabrv,reference}

\newpage
\clearpage
\setcounter{page}{1}


\begin{appendices}

\section{Proof of Proposition \ref{proposition_1}}\label{proof_proposition_1}
When $y_k = 0$, constraint~\eqref{const_p2_3} in $\mathcal{P}2$ enforces $\hat{z}_{k,l_i} = 0$, which aligns with the computing resource allocation in $\mathcal{P}1$ under the same condition. When $y_k > 0$, the model downloading latency becomes constant, and the E2E latency depends solely on $\hat{z}_{k,l_i}$. Therefore, solving $\mathcal{P}2$ under the energy constraint to minimize the E2E latency yields the same computing resource allocation as in $\mathcal{P}1$ when $y_{k}$ and the provisioned model $i$ are given. This completes the proof.

\section{Proof of Proposition \ref{lemma_0}}\label{proof_lemma_0}
Suppose in the optimal solution $\hat{{\bf{Z}}}^{*}_{k,i}$ to $\mathcal{P}2$ under $y_{k}\ne 0$ and $e_{k,i}\left({\bf{1}}\right)> Q_{k}$, there exists a layer $l_{i}\in\left[2,L_{i}\right]$ such that $t_{k,l_{i}-1}\left(y_{k},\hat{{\bf{Z}}}^{*}_{k,i}\right)< \sum\limits_{l'_{i}=1}^{l_{i}}\tau_{k,l'_{i}}$. Since $T_{k,i}\left(z_{k,l_{i}}\right)$ monotonically increases as $z_{k,l_{i}}$ decreases, we can update $\hat{{\bf{Z}}}^{*}_{k,i}$ to $\hat{{\bf{Z}}}'^{*}_{k,i}$ by decreasing $\hat{z}^{*}_{k,l_{i}-1}$ to $\hat{z}'^{*}_{k,l_{i}-1}$, such that $t_{k,l_{i}-1}\left(y_{k},\hat{{\bf{Z}}}'^{*}_{k,i}\right)= \sum\limits_{l'_{i}=1}^{l_{i}}\tau_{k,l'_{i}}$. This update reduces energy consumption without increasing either the inference start time of layer $l_{i}$ or $t_{k,L_{i}}\left(y_{k},\hat{{\bf{Z}}}'^{*}_{k,i}\right)$. However, the saved energy could be reallocated to increase $\hat{z}^{*}_{k,l'_{i}}$ for layer $l'_{i}$ with $t_{k,l'_{i}}\left(y_{k},\hat{{\bf{Z}}}^{*}_{k,i}\right)>\sum\limits_{l''_{i}=1}^{l'_{i}+1}\tau_{k,l''_{i}}$, thus reducing $t_{k,l'_{i}}\left(y_{k},\hat{{\bf{Z}}}'^{*}_{k,i}\right)$, thus decreasing $t_{k,L_{i}}\left(y_{k},\hat{{\bf{Z}}}'^{*}_{k,i}\right)$. This contradicts the optimality of $\hat{{\bf{Z}}}^{*}_{k,i}$, which completes the proof. 

\section{Proof of Proposition \ref{lemma_2}}\label{proof_lemma_2}
When $y_{k}=0$, \eqref{const_p2_3} implies that $\hat{z}^{*}_{k,l_{i}}=0$, which corresponds to the first case in \eqref{eq_z_2}. When $y_{k}\ne0$ and $e_{k,i}\left({\bf{1}}\right)\le Q_{k}$, this indicates that setting $\hat{z}_{k,l_{i}}=1$ for all layers satisfies \eqref{const_p2_1}. Since the objective function of $\mathcal{P}2$ is non-increasing as $\hat{z}_{k,l_{i}}$ grows, this setting preserves the optimality, corresponding to the second case in \eqref{eq_z_2}. 

Next, we derive $\hat{z}^{*}_{k,l_{i}}$ when $y_{k}\ne 0$ and $e_{k,i}\left({\bf{1}}\right)> Q_{k}$. Leveraging \eqref{eq_T} and Proposition~\ref{lemma_0}, we denote the inference start time of layer $l_{i}\in\left[1,L_{i}\right]$ of user $k$ by 
\begin{equation}\label{eq_s_l}
    s_{k,l_{i}}=
    \begin{cases}
        s_{k,1} = \max\left\{\tau_{k,1},T_{k,i}\left(\hat{z}_{k,0}\right)\right\}, \text{ if }l_{i}=1,\\
        s_{k,1}+\sum\limits_{l'_{i}=1}^{l_{i}-1}T_{k,i}\left(\hat{z}_{k,l'_{i}}\right),\text{ if }l_{i}\in\left[2,L_{i}\right].
    \end{cases}
\end{equation}
Based on \eqref{eq_T} and \eqref{eq_s_l},  $t_{k,l_{i}}\left(y_{k},\hat{{\bf{Z}}}_{k,i}\right)$ can be equivalently expressed as $t_{k,l_{i}}\left(y_{k},\hat{{\bf{Z}}}_{k,i}\right)=s_{k,l_{i}}+T_{k,i}\left(\hat{z}_{k,l_{i}}\right)=s_{k,1}+\sum\limits_{l'_{i}=1}^{l_{i}}\left[V_{1}\left(k,l'_{i}\right)+\frac{\Gamma_{k,l'_{i}}}{\hat{z}_{k,l'_{i}}}\right]$, where $\Gamma_{k,l_{i}}=\frac{b_{k}W_{l_{i}}\kappa_{k}}{f_{k}}$. Substituting this into $\mathcal{P}2$, when $y_{k}>0$ and $e_{k,i}\left({\bf{1}}\right)> Q_{k}$, $\mathcal{P}2$ can be equivalently reformulated as
\begin{subequations}
	\begin{equation}
		{\mathcal{P}2'}:\ \mathop{\min}\limits_{\hat{{\bf{Z}}}_{k,i}}\ s_{k,1}+\sum\limits_{l_{i}=1}^{L_{i}}\left[V_{1}\left(k,l_{i}\right)+\frac{\Gamma_{k,l_{i}}}{\hat{z}_{k,l_{i}}}\right]
	\end{equation}	
        \begin{equation}\label{const_p4_2}
		{\rm{s.t.}} \ \sum\limits_{l_{i}=1}^{L_{i}}\Gamma_{k,l_{i}}\hat{z}^{2}_{k,l_{i}}\le Q'_{k}, 
	\end{equation}	
        \begin{equation}\label{const_p4_1}
		s_{k,1}+\sum\limits_{l'_{i}=1}^{l_{i}}\left[V_{1}\left(k,l'_{i}\right)+\frac{\Gamma_{k,l'_{i}}}{\hat{z}_{k,l'_{i}}}\right]\ge \sum\limits_{l_{i}'=1}^{l_{i}+1}\tau_{k,l'_{i}},   \forall l_{i}\in\left[1, L_{i}-1\right] ,
	\end{equation}	
        \begin{equation}\label{const_p4_3}
		0\le \hat{z}_{k,l_{i}}\le 1,\ \forall l_{i}\in\left[1, L_{i}\right] ,
	\end{equation}	
\end{subequations}   
where $Q'_{k}=\frac{Q_{k}-e_{1}\left(k,i\right)}{\Psi_{k}f^{3}_{k}}$. Note that $\mathcal{P}2'$ is a convex problem, as the objective function and constraints of $\mathcal{P}2'$ are convex functions of $\hat{z}_{k,l_{i}}$.

The partial Lagrangian function of $\mathcal{P}2'$ is given by
\begin{equation}
    \begin{aligned}
    \mathcal{L}_{k,i}
    &=s_{k,1}+\sum\limits_{l_{i}=1}^{L_{i}}V_{1}\left(k,l_{i}\right)+\sum\limits_{l_{i}=1}^{L_{i}}\frac{\Gamma_{k,l_{i}}}{\hat{z}_{k,l_{i}}}\\
    &+\sum\limits_{l_{i}=1}^{L_{i}-1}\mu_{l_{i}}\left[ \sum\limits_{l_{i}'=1}^{l_{i}+1}\tau_{k,l'_{i}}-s_{k,1}-\sum\limits_{l'_{i}=1}^{l_{i}}V_{1}\left(k,l'_{i}\right)-\sum\limits_{l'_{i}=1}^{l_{i}}\frac{\Gamma_{k,l'_{i}}}{\hat{z}_{k,l'_{i}}}\right]\\
    &+\eta\left[\sum\limits_{l_{i}=1}^{L_{i}}\Gamma_{k,l_{i}}\hat{z}^{2}_{k,l_{i}}-Q'_{k}\right].
    \end{aligned}
\end{equation}
Letting $\mu_{L_{i}}=0$, $\frac{\partial\mathcal{L}_{k,i}}{\partial{\hat{z}_{k,l_{i}}}}$ can be expressed as
\begin{equation}
    \frac{\partial\mathcal{L}_{k,i}}{\partial{\hat{z}_{k,l_{i}}}}
=-\frac{\Gamma_{k,l_{i}}}{\hat{z}^{2}_{k,l_{i}}}+2\eta \Gamma_{k,l_{i}}\hat{z}_{k,l_{i}}+\frac{\Gamma_{k,l_{i}}}{\hat{z}^{2}_{k,l_{i}}}\sum\limits_{l'_{i}=l_{i}}^{L_{i}}\mu_{l'_{i}},\ 1\le l_{i} \le L_{i}.
\end{equation}

We first derive the relationships between the optimal solution $\hat{z}^{*}_{k,l_{i}}$ to $\mathcal{P}2'$ and the corresponding optimal Lagrange multipliers $\mu^{*}_{l_{i}}$ and $\eta^{*}$. 
The stationary condition $\frac{\partial\mathcal{L}_{k,i}}{\partial{\hat{z}^{*}_{k,l_{i}}}}=0$ in the Karush-Kuhn-Tucker (KKT) conditions yields
\begin{equation}\label{eq_kkt_0}
    -\frac{\Gamma_{k,l_{i}}}{\left(\hat{z}^{*}_{k,l_{i}}\right)^{2}}+2\eta^{*} \Gamma_{k,l_{i}}\hat{z}^{*}_{k,l_{i}}+\frac{\Gamma_{k,l_{i}}}{\left(\hat{z}^{*}_{k,l_{i}}\right)^{2}}\sum\limits_{l'_{i}=l_{i}}^{L_{i}}\mu^{*}_{l'_{i}}=0,\ 1\le l_{i} \le L_{i}.
\end{equation}
On the one hand, if $\eta^{*} =0$, then from \eqref{eq_kkt_0}, we can derive that  $\sum\limits_{l'_{i}=l_{i}}^{L_{i}}\mu^{*}_{l'_{i}}=1, \ 1\le l_{i} \le L_{i}$. Moreover, since the objective function of $\mathcal{P}2'$ is monotonically decreasing as $\hat{z}_{k,l_{i}}$ grows, based on \eqref{const_p4_1}, $\hat{z}^{*}_{k,l_{i}}$ is given by
\begin{equation}
    \hat{z}^{*}_{k,l_{i}}=\dot{z}_{k,l_{i}}\triangleq
    \begin{cases}
        \frac{\Gamma_{k,1}}{\tau_{k,1}+\tau_{k,2}-s_{k,1}-V_{1}\left(k,1\right)}, \text{ if } l_{i}=1,\\
        \frac{\Gamma_{k,l_{i}}}{\tau_{k,l_{i}+1}-V_{1}\left(k,l_{i}\right)},\text{ if } 2\le l_{i}\le L_{i}-1,\\
        1, \text{ if } l_{i}=L_{i},
    \end{cases}
\end{equation}
when $\sum\limits_{l_{i}=1}^{L_{i}}\Gamma_{k,l_{i}}\dot{z}^{2}_{k,l_{i}}\le Q'_{k}$ and $0\le\dot{z}_{k,l_{i}}\le 1$ hold. This corresponds to the third case in \eqref{eq_z_2}. 
On the other hand, if $\eta^{*}\ne0$, then based on \eqref{eq_kkt_0} and \eqref{const_p4_3}, $\hat{z}^{*}_{k,l_{i}}$ is given by
\begin{equation}\label{eq_z_1}
    \hat{z}^{*}_{k,l_{i}}=\min\left\{1,\max\left\{0,\sqrt[3]{\frac{\rho^{*}_{l_{i}}}{2\eta^{*}}}\right\}\right\},
\end{equation}
where $\rho^{*}_{l_{i}}=1-\sum\limits_{l'_{i}=l_{i}}^{L_{i}}\mu^{*}_{l'_{i}}$, which corresponds to the fourth case in \eqref{eq_z_2}.

Second, we 
derive the relationship between $\mu^{*}_{l_{i}}$ and $\eta^{*}$ in \eqref{eq_z_1} for the fourth case in \eqref{eq_z_2}. 
Based on the complementary slackness condition
\begin{equation}\label{eq_eta_kkt}
\eta^{*}\left[\sum\limits_{l_{i}=1}^{L_{i}}\Gamma_{k,l_{i}}\left(\hat{z}^{*}_{k,l_{i}}\right)^{2}-Q'_{k}\right]=0,
\end{equation}
in the KKT conditions, it follows that
\begin{equation}\label{eq_kkt_energy}
    \sum\limits_{l_{i}=1}^{L_{i}}\Gamma_{k,l_{i}}\left(\hat{z}^{*}_{k,l_{i}}\right)^{2}=Q'_{k}.
\end{equation}
Then, by substituting \eqref{eq_z_1} into \eqref{eq_kkt_energy}, we obtain
\begin{equation}\label{eq_eta_0}
    \sum\limits_{l_{i}=1}^{L_{i}}\Gamma_{k,l_{i}}\min\left\{1,\max\left\{0,\left(\frac{\rho^{*}_{l_{i}}}{2\eta^{*}}\right)^{\frac{2}{3}}\right\}\right\}=Q'_{k},
\end{equation}
from which $\eta^{*}$ can be determined as
\begin{equation}\label{eq_eta_1}
    \eta^{*}=\chi\left(\mu^{*}_{l_{i}},Q'_{k}\right),
\end{equation}
where $\chi\left(\mu^{*}_{l_{i}},Q'_{k}\right)$ is the inverse function that computes $\eta^{*}$ from \eqref{eq_eta_0} given $\bigcup\limits_{l_{i}=1}^{L_{i}}\mu^{*}_{l_{i}}$ and $Q'_{k}$. 

Finally, by leveraging \eqref{eq_z_1} and \eqref{eq_eta_1}, and applying the projected subgradient method~\cite{1664999}, 
we derive the update rules for $\mu^{\left(m\right)}_{l_{i}}$, $\eta^{\left(m\right)}$, and $\hat{z}^{\left(m\right)}_{k,l_{i}}$ in the $m$-th iteration for the fourth case in \eqref{eq_z_2}, as shown in \eqref{eq_update}. Since $\mathcal{P}2'$ is a convex problem and satisfies the Slater condition, the global optimal solution can be obtained upon convergence of the update rules in \eqref{eq_update} for this case. This completes the proof.

\section{Proof of Corollary \ref{remark_2}}\label{proof_remark_2}
On the one hand, if $e_{k,i}\left({\bf{1}}\right)\le Q_{k}$, then according to \eqref{eq_z_2}, setting the GPU frequency allocation scaling factor to 1 for all layers of model $i$ preserves the optimality of $\mathcal{P}1$. On the other hand, if both $e_{k,i}\left({\bf{1}}\right)$ and $e_{k,i}\left(\dot{{\bf{Z}}}_{k,i}\right)$ exceed $Q_{k}$, then, from \eqref{eq_eta_kkt}, we can derive that $\eta^{*}\ne 0$, where $\eta^{*}$ denotes the optimal Lagrange multiplier of $\mathcal{P}2'$ as defined in the proof of Proposition \ref{lemma_2}. According to Proposition~\ref{proposition_1}, the solution to $\mathcal{P}2'$ coincides with the computing resource allocation in $\mathcal{P}1$ when $y_{k}$ and the provisioned model $i$ are given. Therefore, based on \eqref{eq_kkt_energy}, we conclude that user $k$ consumes the entire energy budget to perform forward propagation with model $i$. This completes the proof.

\section{Proof of Theorem \ref{theorem_3}}\label{proof_theorem_3}
The proof of the optimality of Algorithm~\ref{algorithm_greedy} proceeds by analyzing its three key components: (i) calculating the minimum feasible bandwidth allocation $\check{y}_{k}$ using Algorithm~\ref{algorithm_bi}; (ii) selecting users in ascending order of $\check{y}_{k}$ within the while loop; and (iii) assigning $\check{y}_{k^{*}}$, $\check{x}^{*}_{k^{*},i}$, and $\check{z}^{*}_{k^{*},l_{i}}$ to $y^{*}_{k^{*}}$, $x^{*}_{k^{*},i}$, and $z^{*}_{k^{*},l_{i}}$, respectively in Line~\ref{line:greedy_z}.
    
    We begin by proving that Algorithm~\ref{algorithm_bi} guarantees that the produced $\check{y}_{k}$ is the minimum feasible bandwidth for user $k$, and the corresponding $\check{x}^{*}_{k,i}$ and $\check{z}^{*}_{k,l_{i}}$ are optimal under $\check{y}_{k}$. First, from Proposition~\ref{proposition_1}, for any given $y_{k}$, the optimal computing resource allocation for user $k$ with model $i$ and the corresponding minimum E2E latency $t_{k,L_{i}}\left(y_{k},\hat{{\bf{Z}}}^{*}_{k,i}\right)$ can be obtained by solving $\mathcal{P}2$. Second, Proposition~\ref{theorem_1} ensures that for any given $y_{k}$, $t_{k,L_{i^{*}}}\left(y_{k},\hat{\bf{Z}}^{*}_{k,i^{*}}\right)$ is the minimum E2E latency of user $k$ with the provisioned model $i^{*}$. Third, from \eqref{eq_final_e2e}, $t_{k,L_{i}}\left(y_{k},\hat{\bf{Z}}^{*}_{k,i^{*}}\right)$ is a non-increasing function of $y_{k}$. Specifically, as $y_{k}$ grows, the total model downloading latency decreases. Moreover, due to the parallelization between the model downloading and inference, the E2E latency does not increase as $y_{k}$ grows since the computing resource allocation under a larger $y_{k}$ results in no greater inference latency than that under a smaller $y_{k}$. Therefore, by leveraging the bisection search method, Algorithm~\ref{algorithm_bi} can obtain $\check{y}_{k}$ by iteratively narrowing the feasible interval until $t_{k,L_{i^{*}}}\left(\check{y}_{k},\hat{{\bf{Z}}}^{*}_{k,i^{*}}\right)\le \bar{T}_{k}$ is satisfied within a desired error bound.
    
    Next, we prove that selecting users in ascending order of $\check{y}_{k}$ in Algorithm~\ref{algorithm_greedy} preserves the optimality of the solution to $\mathcal{P}1$. Suppose that the served user set $\mathcal{K}^{*}$ produced by Algorithm \ref{algorithm_greedy} is not optimal. Then, there must exist another solution yielding a higher task throughput, with the served user set denoted as $\dot{\mathcal{K}}$, where $\left|\dot{\mathcal{K}}\right| > \left|\mathcal{K}^{*}\right|$, and $\dot{\mathcal{K}}$ does not exactly consist of the $\left|\dot{\mathcal{K}}\right|$ users with the smallest values of $\check{y}_{k}$. Consider updating the set $\dot{\mathcal{K}}$ by replacing its users with the first $\left|\dot{\mathcal{K}}\right|$ users with the smallest values of $\check{y}_{k}$, and denote the updated user set as $\dot{\mathcal{K}}'$. Clearly, $\sum\limits_{k\in\dot{\mathcal{K}}'}\check{y}_{k}\le1$ still holds. However, the replacement results in $\sum\limits_{k\in\dot{\mathcal{K}}'}\check{y}_{k}\le \sum\limits_{k\in\dot{\mathcal{K}}}\check{y}_{k}$, indicating that more users could be allocated bandwidth and be served after the replacement. This contradicts the optimality of $\dot{\mathcal{K}}$ and completes the proof that selecting users in ascending order of $\check{y}_{k}$ preserves the optimality. 
    
    At last, since $\check{y}_{k^{*}}$ is the minimum feasible bandwidth allocation for user $k^{*}$ to complete the inference task within $\bar{T}_{k^{*}}$, the value of $\check{y}_{k^{*}}$ is the optimal bandwidth allocation for user $k^{*}$. Moreover, from Proposition~\ref{theorem_1}, $\check{x}^{*}_{k^{*},i}$ and $\check{z}^{*}_{k^{*},l_{i}}$ are the optimal model provisioning and computing resource allocation, respectively, under $\check{y}_{k^{*}}$. Therefore, the assignments in Line~\ref{line:greedy_z} guarantee the optimal bandwidth allocation, model provisioning, and computing resource allocation for the selected user $k^{*}$. This completes the proof.

\section{Proof of Theorem \ref{theorem_4}}\label{proof_theorem_4}
We begin by analyzing the time complexity of Algorithm~\ref{algorithm_bi}. First, the complexity of Line~\ref{line:bi_solve_p2} in Algorithm~\ref{algorithm_bi} is $O\left(\frac{L_{i}}{\hat{\epsilon}^{2}}\right)$, where 
    $\hat{\epsilon}$ is the desired error bound for deciding the optimal Lagrange multipliers \cite{nesterov2013introductory}. 
Therefore, the time complexity of Lines~\ref{line:bi_for_start} to \ref{line:bi_for_end} is $O\left(\sum\limits_{i\in\mathcal{I}_{k}}\frac{L_{i}}{\hat{\epsilon}^{2}}\right)$. Second, in each iteration of Algorithm~\ref{algorithm_bi}, Line~\ref{line:bi_i_*} involves checking all $\left|\mathcal{I}_{k}\right|$ models, resulting in time complexity $O\left(\left|\mathcal{I}_{k}\right|\right)$. Besides, the complexity of Lines~\ref{line:bi_if_start} to \ref{line:bi_if_end} is $O\left(1\right)$. Third, since Algorithm~\ref{algorithm_bi} employs the bisection search method, the while loop from Line \ref{line:bi_while_start} to Line~\ref{line:bi_while_end} requires $O\left({\rm{log}}\frac{1}{\epsilon}\right)$ iterations to converge within the desired error bound $\epsilon$. 
Finally, the complexity of Line~\ref{line:bi_tilde_x_z} is $O\left(1\right)$ since $\hat{z}^{*}_{k,l_{i}}$ and $i^{*}$ have already been obtained in Line~\ref{line:bi_solve_p2} and \ref{line:bi_i_*}, respectively. Therefore, the complexity of Lines~\ref{line:bi_if_i_start} to \ref{line:bi_if_i_end} is $O\left(1\right)$. 
Based on the above analysis, the total time complexity of Algorithm~\ref{algorithm_bi} is $O\left({\rm{log}}\frac{1}{\epsilon}\left(\sum\limits_{i\in\mathcal{I}_{k}}\frac{L_{i}}{\hat{\epsilon}^{2}}+\left|\mathcal{I}_{k}\right|\right)\right)=O\left(\frac{1}{\hat{\epsilon}^{2}}{\rm{log}}\frac{1}{\epsilon}\sum\limits_{i\in\mathcal{I}_{k}}L_{i}\right)=O\left(\sum\limits_{i\in\mathcal{I}_{k}}L_{i}\right)$.

The time complexity of Algorithm~\ref{algorithm_greedy} is derived as follows. 
First, based on the time complexity of Algorithm~\ref{algorithm_bi}, that of Line~\ref{line:greedy_call} in Algorithm~\ref{algorithm_greedy} is $O\left(\sum\limits_{k\in\mathcal{K}}\sum\limits_{i\in\mathcal{I}_{k}}L_{i}\right)$. Second, in each while loop iteration, Line~\ref{line:greedy_k} involves checking all $K$ users, yielding a time complexity of $O\left(K\right)$. Third, the complexity of both Line~\ref{line:greedy_z} and \ref{line:greedy_u} is $O\left(1\right)$. At last, Algorithm~\ref{algorithm_greedy} requires at most $O\left(K\right)$ iterations in the while loop from Line~\ref{line:greedy_while_start} to Line~\ref{line:greedy_while_end}. 
Therefore, the total time complexity of Algorithm~\ref{algorithm_greedy} is $O\left(\sum\limits_{k\in\mathcal{K}}\sum\limits_{i\in\mathcal{I}_{k}}L_{i}+\sum\limits_{k\in\mathcal{K}}K\right)\le O\left(K^{2}+KIL_{\max}\right)$, where $L_{\max}=\mathop{\max}\limits_{i \in \mathcal{I}}\left\{L_{i}\right\}$. This completes the proof.
\end{appendices}

\end{document}